\documentclass[journal,10pt]{IEEEtran}
\usepackage{mathrsfs}
\usepackage{bbding}
\usepackage{amsfonts}
\usepackage{amsmath}
\usepackage{graphicx}
\usepackage{subfigure}
\usepackage{latexsym}
\usepackage{amssymb}
\usepackage{amsmath}
\usepackage{enumerate}
\usepackage{color}

\newtheorem{Assumption}{Assumption}
\newtheorem{Algorithm}{Algorithm}

\newtheorem{Definition}{Definition}
\newtheorem{Lemma}{Lemma}
\newtheorem{Problem}{Problem}
\newtheorem{Remark}{Remark}
\newtheorem{Theorem}{Theorem}
\newtheorem{Corollary}{Corollary}

\newcommand{\bs}{\begin{small}}
\newcommand{\bsc}{\end{small}}
\usepackage[font=small]{caption}
\newcommand{\txtblue}{\textcolor{black}}

\allowdisplaybreaks[4]

\begin{document}
\title{Cross-Layer MIMO Transceiver Optimization  for Multimedia Streaming in  Interference Networks}

\author
{\IEEEauthorblockN{Fan Zhang, \emph{StMIEEE}, Vincent K. N. Lau, \emph{FIEEE}} \thanks{Fan Zhang and Vincent K. N. Lau are with Department of Electronic and Computer Engineering, Hong Kong University of Science and Technology, Hong Kong.}
} 
\maketitle

\begin{abstract}
In this paper, we consider dynamic precoder/decorrelator optimization for multimedia streaming in MIMO interference networks.  We propose a truly cross-layer framework in the sense that the optimization objective is the application level  performance metrics for multimedia streaming, namely the {\em playback interruption} and {\em buffer overflow} probabilities. The optimization variables are the MIMO precoders/decorrelators at the transmitters and the receivers, which are adaptive to both the instantaneous channel condition and the playback queue length. The problem is a challenging multi-dimensional stochastic optimization problem and brute-force solution has exponential complexity. By exploiting the underlying timescale separation and special structure in the problem, we derive a closed-form approximation of the value function based on continuous time perturbation.  Using this approximation, we propose a low complexity dynamic MIMO precoder/decorrelator control algorithm by solving an equivalent weighted MMSE problem. We also establish the technical conditions for asymptotic optimality of the low complexity control algorithm.  Finally, the proposed scheme is compared  with various baselines through simulations and it is shown that significant  performance gain can be achieved.
\end{abstract}

\section{introduction}
\subsubsection{Background}
There is a surge of interest in multimedia streaming in wireless systems and high quality real-time multimedia streaming applications pose great challenges to the design of the future wireless systems. In this paper, we consider multimedia streaming in MIMO interference networks where multiple BSs simultaneously deliver multimedia data to their associated mobile users over a shared wireless link. The performance of the  interference network is fundamentally limited by the inter-cell interference from the cross links. There are many existing works on the interference mitigation for MIMO interference networks. In \cite{IA1}, \cite{IA2}, the authors show that interference alignment  can achieve optimal degrees of freedom  of a $K$-user interference network using infinite dimension time or frequency symbol extension.  In \cite{optMISO1}, \cite{optMISO2}, the authors consider joint beamforming to minimize the sum  mean squared error (MSE) or the transmit power of a multi-user MIMO system using optimization approaches. In \cite{decenMISO1}, \cite{decenMISO2}, the authors analyze the achievable rate region of a multi-antenna interference channel from a game-theoretic perspective and  consider a distributed  beamforming design  using non-cooperative game. However, these solution frameworks are not {\em truly cross-layer design}  \cite{crosslayer1}, \cite{crosslayer2}, in the sense that the optimization objectives are the physical layer metrics (e.g., throughput, SNR), which may not be directly related to the application level performance metrics in multimedia streaming. Furthermore, the resulting control policy is adaptive to the channel state information (CSI) only, which exploits good transmission opportunities from the time-varying physical channels. However, for real-time multimedia streaming, dynamic control policy adaptive to the  instantaneous queue length (QSI) is also very important because they give information about the {\em urgency} of the data flows. 

Control policy adaptive to both the CSI and the QSI is very challenging because the associated optimization problem belongs to an infinite dimension stochastic optimization problem. A systematic approach is to formulate the problem into a Markov Decision Process (MDP) \cite{Cao}, \cite{DP_Bertsekas}. In \cite{mdpapp1}, \cite{mdpapp2}, delay minimization using the MDP approach is considered. There are also a number of works \cite{lya1}, \cite{lya2} that adopt the stochastic Lyapunov optimization technique for average delay minimization of wireless networks. However, these techniques cannot be easily used for multimedia streaming applications because average delay is not the end--to--end performance metric for multimedia streaming. For multimedia streaming applications, there is a playback buffer at the each mobile user and \txtblue{the \emph{playback interruption probability} and the \emph{buffer overflow probability} are the two important end--to--end performance metrics\footnote{\txtblue{Note that if we want to have good end-to-end performance for an application, we need to take the end-to-end performance metrics into the design considerations directly, instead of optimizing some intermediate performance metrics (such as weighted MMSE or sum rate).}}}. Playback interruption occurs when the playback buffer underflows and this is highly undesirable for the end user experience. On the other hand, due to the finite buffer size nature in  practical systems,  new packet arrivals will be dropped when playback buffer is full. This is  also undesirable due to the wastage of wireless resource used to transmit these dropped packets. In \cite{mdpsett2}, \cite{mdpsett1}, the authors consider a fully dynamic power control and rate adaptation for video streaming over a wireless link using MDP. The optimality condition, namely the {\em Bellman equation}, is obtained and solved using conventional value iteration algorithm \cite{Cao}, \cite{DP_Bertsekas}. However, the solution cannot be extended to deal with the multi-flow stochastic problem due to the curse of dimensionality. In our problem, there are $K$ multimedia streaming flows in the system, and the queue dynamics of the $K$ flows are complex-coupled together. This is because the \emph{data rate} of each playback queue at the mobiles depends on the beamforming control actions of the other flows due to the mutual interference. As a result, brute-force value iteration or policy iteration \cite{Cao}, \cite{DP_Bertsekas} will result in solutions with exponential complexity and they will not be viable  in practice. 

\subsubsection{Our Contribution}
In this paper, we consider a truly cross-layer optimization framework for real-time multimedia streaming applications in MIMO interference networks. Unlike many existing works on MIMO precoder/decorrelator optimization, the optimization objectives we consider, namely the {\em playback interruption} and {\em buffer overflow} probabilities, are directly related to the application level performance metrics. Furthermore, the optimization variables are the MIMO precoders/decorrelators which are adaptive to the instantaneous CSI and the instantaneous QSI at the playback buffers. The associated problem belongs to a $K$-dimensional stochastic optimization and brute-force solution \cite{Cao}, \cite{DP_Bertsekas} has exponential complexity. By exploiting the special structure in the problem as well as the timescale separation between the slot duration and the interruption/overflow events, we obtained an {\em equivalent optimality condition} for the MDP in terms of a $K$-dimensional partial differential equation (PDE). The solution of the PDE is called the {\em value functions} and they capture the {\em dynamic urgency} of the $K$ data flows. To deal with the challenge due to the queue coupling and the curse of dimensionality, we derive a closed-form approximate solution for the PDE using perturbation theory. Based on the derived approximate value function, the MIMO precoders/decorrelators are optimized by solving a per-stage weighted MMSE problem \cite{WMMSE}, where the instantaneous QSI affects the weights via the value function. While the per-stage problem is non-convex, we establish technical conditions for the asymptotic optimality of the proposed low complexity solution. Finally, we compare the proposed algorithm with various conventional  beamforming schemes through simulations and show that significant  performance gain can be achieved.

\section{system model}
In this section, we introduce the architecture of the multimedia streaming system  in MIMO interference networks, the physical layer model as well as the playback queue model.

\subsection{Architecture of the Multimedia Streaming System in  MIMO Interference Networks}

Fig.~\ref{K_pair} shows a typical architecture of the multimedia streaming system in MIMO interference networks. The raw multimedia files are pre-compressed and saved in the storage devices in the multimedia streaming server (MSS).  There are $K$ mobile users streaming multimedia files from the MSS via a radio access network (RAN).  Specifically, upon the request from the users, the MSS retrieves the pre-stored multimedia file and transmits it to the users over the RAN.   Each mobile user $k$ \emph{consumes}  the received multimedia packets at a constant playback rate $\mu_k$. Furthermore, the RAN consists of $K$ BSs, which are connected to the MSS via  a high speed backhaul  links. BS $k$ sends information to user $k$. Each BS is equipped with $N_t \geq K$ antennas and each user is equipped with $N_r$   antennas.  All the $K$ BSs share a common spectrum with bandwidth $W$Hz and hence, they potentially interfere with each other. In this paper, the time dimension is partitioned into decision slots indexed by $t$ with  slot duration $\tau$. For example, in LTE \cite{lte}, the physical layer is organized into  radio frames  (corresponding to slot in our problem), and the generic radio frame has a time duration of 10 ms.

\begin{figure}[t]
  \centering
  \includegraphics[width=3.6in]{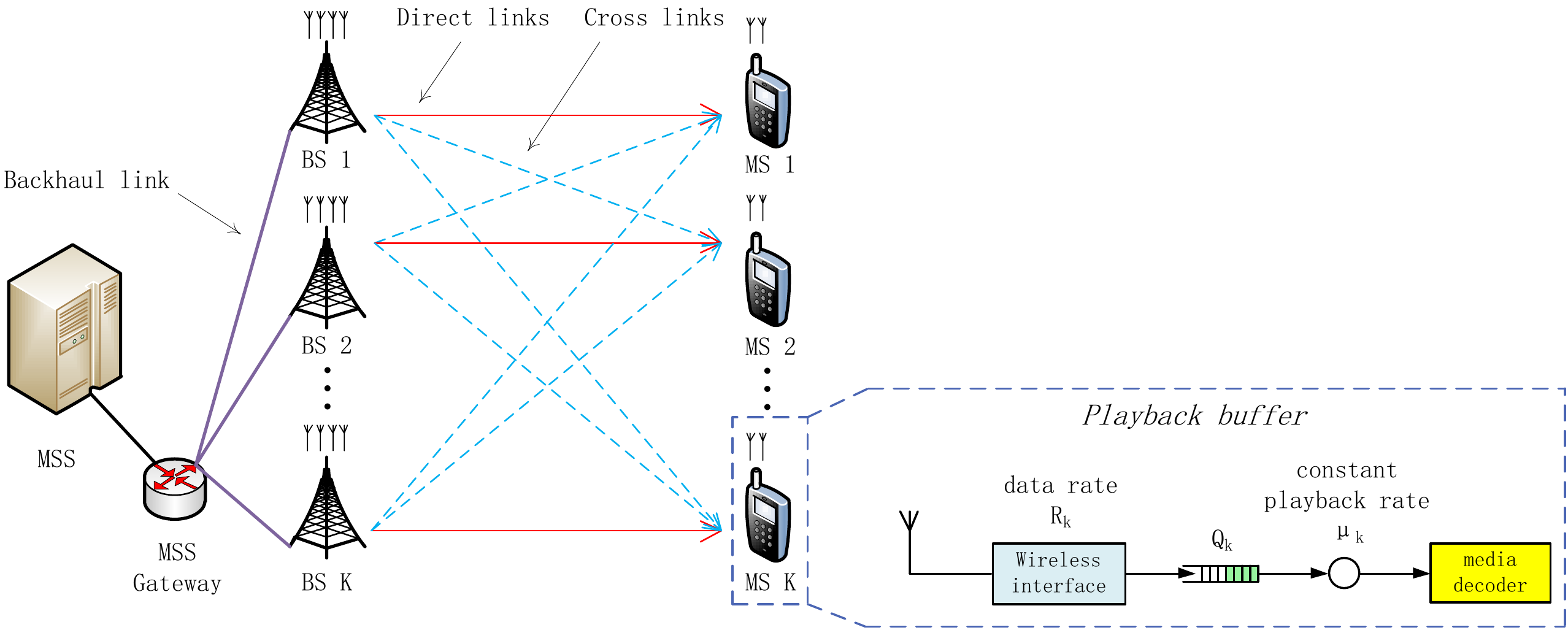}
  \caption{Architecture of a multimedia system in  MIMO interference network.}
  \label{K_pair}
\end{figure}

\subsection{Physical Layer Model}
The RAN and  the $K$ mobile users forms a MIMO interference network and the performance is limited by the inter-cell (cross channel) interference between the BSs. To deal with the interference issue,  joint precoder/decorrelator optimization  \cite{comp2}  is adopted at the BSs. Let $\mathbf{F}_k \in \mathbb{C}^{N_t \times d}$ be the transmit precoding matrix of BS $k$, where $d=\min\left\{N_t, N_r\right\}$ the number of data streams transmitted by each Tx-Rx pair\footnote{In this paper, we shall refer to BS as transmitter (Tx) and mobile users as receivers (Rx),  and each BS and the associated mobile user pair as a Tx-Rx pair.}.  Let $\mathbf{U}_k\in \mathbb{C}^{N_r \times d}$ be the decoding matrix of mobile user $k$. The received signal $\mathbf{y}_k \in \mathbb{C}^{d \times 1}$ at user $k$ is given by
\begin{equation}	\label{recvsig}
	\mathbf{y}_k =  \mathbf{U}_k^\dagger \big( \sqrt{L_{kk}} \mathbf{H}_{kk} \mathbf{F}_k \mathbf{s}_k+ \sum_{j \neq k}   \sqrt{L_{kj}} \mathbf{H}_{kj} \mathbf{F}_j \mathbf{s}_j + \mathbf{n}_k\big)
\end{equation}
where $L_{kj}\in \mathbb{R}^+$ and $\mathbf{H}_{kj}\in \mathbb{C}^{N_r \times N_t}$ are the long-term channel path gain and short-term channel fading matrix from BS $j$ to user $k$, respectively. \txtblue{$\mathbf{s}_k \in \mathbb{C}^{d \times 1}$ is the information symbol for BS $k$ and we assume $\mathbb{E}\big[\mathbf{s}_k \mathbf{s}_k^\dagger\big]=\mathbf{I}$ \cite{WMMSE}.  In practical multimedia streaming applications, the information symbols are drawn from a finite   alphabet constellation set of size $|\mathcal{S}|$, i.e., $\mathbf{s}_k \in \left\{ \boldsymbol{\xi}_k^i\right\}_{i=1}^{|\mathcal{S}|}$  \cite{finitealp} for all $k$.}  $\mathbf{n}_k \sim \mathcal{CN}(0, \mathbf{I})$ is the i.i.d. complex AWGN noise vector. $(\cdot)^\dagger$ represents the conjugate transpose of a matrix. Denote the global CSI as $\mathbf{H}=\left\{\mathbf{H}_{kj}: \forall k, j\right\}$. We have the following assumption on $\mathbf{H}$:
\begin{Assumption}	[Channel Fading Model]	\label{CSIassum}	
$\mathbf{H}_{kj}\left(t\right)$ remains constant within each decision slot  and is  i.i.d. over  slots for all $k,j$. Specifically, each element of $\mathbf{H}_{kj}\left(t\right)$ follows a complex Gaussian distribution with zero mean and unit variance. Furthermore, $\mathbf{H}_{kj}\left(t\right)$ is  independent w.r.t. $k$, $j$. The path gain  $L_{kj}$ remains constant for the duration of the communication session.	~\hfill\IEEEQED
\end{Assumption}

For given CSI $\mathbf{H}$, precoding matrices $\mathbf{F}=\left\{\mathbf{F}_k: \forall k\right\}$ and decoding matrices $\mathbf{U}_k$, the achievable data rate   for the $k$-th Tx-Rx pair (by treating interference as noise) is given by \cite{constellation} 
\begin{align}		
	& R_k\left(\mathbf{H}, \mathbf{F}, \mathbf{U}_k \right) = W \log_2 \det \Big( \mathbf{I}+\txtblue{\zeta} L_{kk}\mathbf{U}_k^\dagger  \mathbf{H}_{kk}\mathbf{F}_k \mathbf{F}_k^\dagger \mathbf{H}_{kk}^\dagger \mathbf{U}_k   \notag \\
	 &\hspace{2.2cm}\Big(\sum_{j \neq k} L_{kj} \mathbf{U}_k^\dagger  \mathbf{H}_{kj}\mathbf{F}_j \mathbf{F}_j^\dagger \mathbf{H}_{kj}^\dagger\mathbf{U}_k  +\mathbf{I} \Big)^{-1}\Big)\label{rate1}
\end{align}\txtblue{where $\zeta \in (0, 1]$ is a constant that is determined by the modulation and coding scheme (MCS) used in the system. For example, $\zeta=0.5$ for QAM constellation at BER= 1\% \cite{constellation} and $\zeta = 1$ for capacity-achieving coding (in which, (\ref{rate1}) corresponds to the instantaneous mutual information). In this paper, our derived results are based on $\zeta=1$ for  simplicity, which can be easily extended to other MCS cases.}

Furthermore, the transmit power for BS $k$ is given by \txtblue{\cite{WMMSE}}
\begin{align}	\label{defpow}
	P_k\left( \mathbf{F}_k\right)=\mathrm{Tr}\left(  \mathbf{F}_k  \mathbf{F}_k^\dagger\right)
\end{align}
where $\mathrm{Tr}\left(  \cdot\right)$ represents the trace operator.

\begin{figure}\centering
  \includegraphics[width=3.5in]{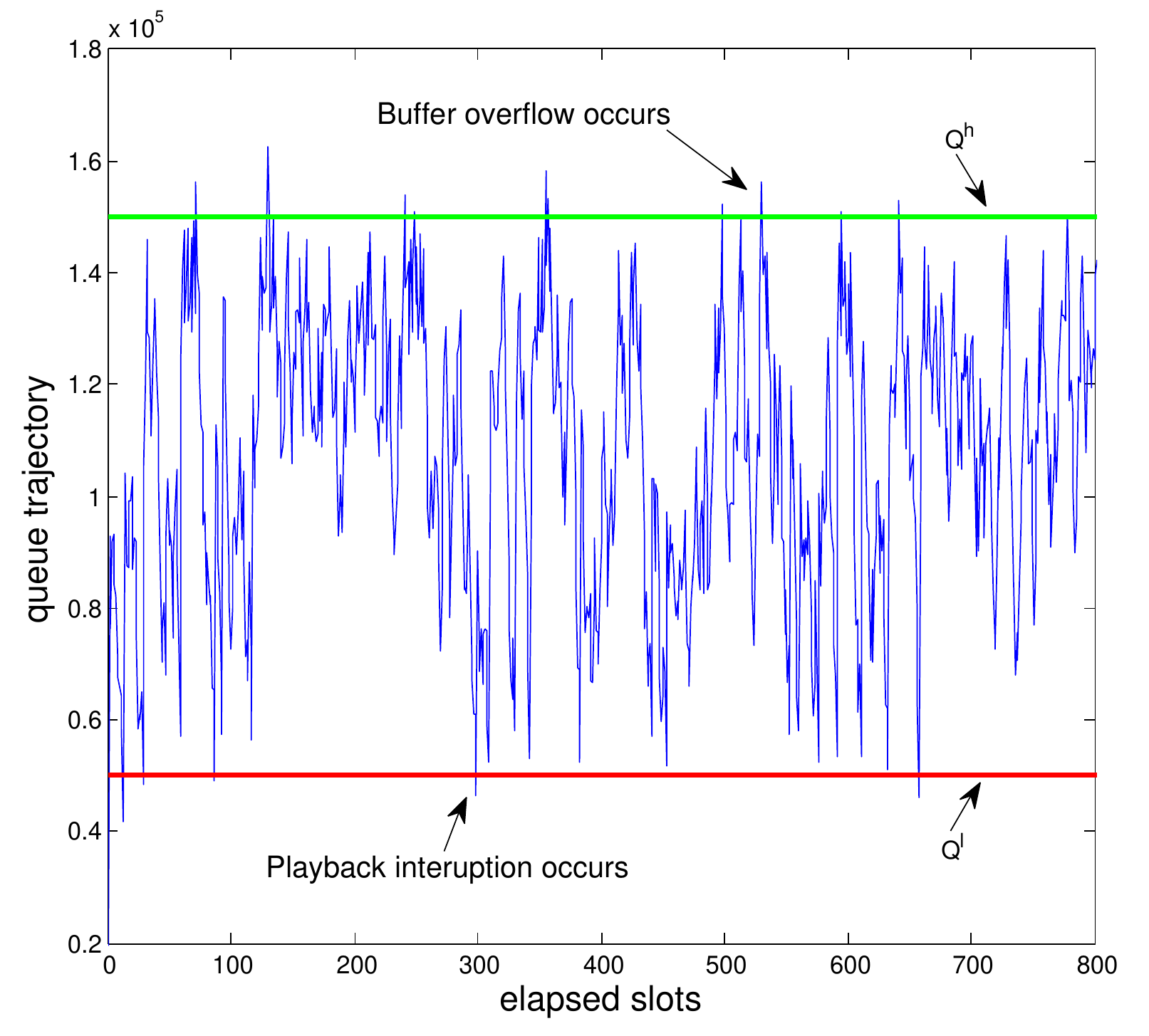}
  \caption{Queue trajectory of the playback queue $Q_k(t)$ at the $k$-th mobile user. The system parameters are configured as in  the simulations in Section  \ref{xxsim}.}
  \label{fig1adsa}
\end{figure}

\subsection{Playback Queue Dynamics at the Mobile Users}	\label{frisss}
As shown in Fig.~\ref{K_pair}, each mobile user maintains a data queue for multimedia playback. Let $Q_k\left( t\right) \in \mathcal{Q}$  denote the QSI (number of bits) at the playback buffer of user $k$ at the beginning of the $t$-th slot, where $\mathcal{Q}=[0, \infty)$ is the   QSI state space. Let $\mathbf{Q} \left( t\right) = \left(Q_1\left(t \right),\dots,Q_K \left( t\right)  \right) \in \boldsymbol{\mathcal{Q}} \triangleq \mathcal{Q}^K$ denote the global QSI. The instantaneous arrivals of the $k$-th playback queue at slot $t$ is given by $R_k(t)\tau$, which is controlled by the precoders $\mathbf{F}(t)$ and decorrelator $\mathbf{U}_k(t)$. The instantaneous departure of the $k$-th playback queue is given by $\mu_k\tau$, which is a constant and depends on the multimedia decoder at the end user.  Hence, the queue dynamics for user $k$ is given by
\begin{align}		
	Q_k(t+1) &= \left[Q_k\left(t\right) - \mu_k  \tau \right]^+  + R_k\left(\mathbf{H}\left( t\right), \mathbf{F}\left( t\right), \mathbf{U}_k (t) \right)  \tau       	\label{Qdyn}
\end{align}	
where $[x]^+=\max\{0,x\}$.  It can be observed that the queue dynamics in the playback buffer is a Markovian queue with \emph{controlled arrivals}.

Fig.~\ref{fig1adsa} illustrates a trajectory of the playback queue $Q_k(t)$ for the $k$-th mobile user. The multimedia files are consumed by the user at  a constant rate $\mu_k\tau$ during each slot. The system has to control the precoders $\mathbf{F}$ and decorrelators  $\mathbf{U}$, so that $Q_k(t)$ will seldom  go beyond certain level (i.e., the green line, which results in buffer overflow) or go below certain level (i.e., the red line, which results in playback interruption).

\begin{Remark} [Coupling Property of  Queue Dynamics] 	\label{coup_rem}
	 The $K$ queue dynamics in the MIMO interference network are coupled together due to the  interference  in (\ref{rate1}). Specifically,  the data rate $R_k$ of each Tx-Rx pair $k$  depends on the precoding matrices  $\left\{\mathbf{F}_j: \forall j \neq k\right\}$ of all the other Tx-Rx pairs. Furthermore, the cross channel path gain $\left\{L_{kj}: \forall k, j, j\neq k \right\}$ measures the coupling intensity  in the interference network.~\hfill\IEEEQED
\end{Remark}

We have the following assumption on the  interference network:
\begin{Assumption}	[Weak Interference Network]
	For each Tx-Rx pair $k$, we assume the long-term cross channel path gains are much  smaller than the direct channel path gain, i.e., $L_{kj} \ll L_{kk}$, $\forall j \neq k$. Furthermore, denote $L=\max\left\{L_{kj}: \forall k, j, k \neq j \right\}$ to be the largest (worst-case) cross channel path gain in the interference network.~\hfill\IEEEQED
\end{Assumption}

The assumption on the weak interference network can be justified in many applications. For example, due to the MAC filtering effect in some protocols, such as CSMA/CA \cite{csmamac}, the interference in the cross channels cannot be too strong. \txtblue{The basic principle of the CSMA/CA is listen-before-talk \cite{csmamac}, which is used to avoid collisions between simultaneous transmissions of the BSs in the neighboring cells. As a result, the MAC protocol determines the subset of the BSs in which the BSs can transmit data simultaneously without causing excessive interference.}  Suppose each BS uses a CSMA/CA MAC protocol with carrier sensing distance $\delta$, then the worst-case path gain between two interfering BSs is given by\footnote{Here we adopt the Friis  path loss model with exponent of 4 \cite{kangshin}, which corresponds to the environment with obstructings in buildings. Note that the results of this paper can be extended easily for other path loss models.} \cite{kangshin}: $L= G^r G^t \left(\frac{ \lambda}{4 \pi  }\right)^2\frac{1}{\delta^4}$, where  $G^r$ and $G^t$ are the receive and transmit antenna gains respectively, and $\lambda$ is the carrier wavelength. \txtblue{For instance, in  IEEE 802.11g \cite{csmajustfy}, the CSMA/CA sensing threshold is around -95 dBm, which corresponds to a sensing distance (i.e., $\delta$) of around 188 m for the indoor environment. To support a 54 Mbps data rate, the receive sensitivity is around -75 dBm. Therefore, such a choice of carrier sensing distance corresponds to a worst-case cross channel path gain of at least 20 dB less than the direct channel path gain.} We shall exploit this weak interference coupling property in Section \ref{weakintder} to derive a closed-form approximate solution to the multi-dimensional MDP problem.

\section{stochastic precoder and decorrelator  control problem formulation}
In this section, we  define the precoder and decorrelator  control policy and formulate the stochastic control problem for multimedia streaming in the MIMO interference network.

\subsection{MIMO Precoder and Decorrelator  Control Policy}	
For notation convenience, we  denote $\boldsymbol{\chi}=\big(\mathbf{H}, \mathbf{Q} \big)$ as the global system state.  At the beginning of  each decision slot, the controller determines the precoders $\mathbf{F}= \left\{\mathbf{F}_k:\forall k\right\}$  and decorrelators $\mathbf{U}= \left\{\mathbf{U}_k: \forall k\right\}$ to minimize the playback interruption and buffer overflow probabilities  of the multimedia streaming applications based on the global  system state $\boldsymbol{\chi}$ according to the following  stationary control policy:

\begin{Definition} \label{deff1} \emph{(Stationary  Precoder and Decorrelator   Control Policy)}	
	A stationary  precoder and decorrelator    control policy $\Omega_k$ for Tx-Rx pair  $k$  is a mapping from the global  system state $\boldsymbol{\chi}$ to the  precoding matrix of BS $k$ and decoding matrices of  user $k$. Specifically, we have $\Omega_k\big(\boldsymbol{\chi}\big)=\left\{ \mathbf{F}_k \in \mathbb{C}^{N_t \times d}, \mathbf{U}_k \in \mathbb{C}^{N_r \times d}\right\}$. Furthermore, let $\Omega= \{ \Omega_k:\forall  k \}$ denote the aggregation of the control policies for all the $K$ BSs.~\hfill\IEEEQED
\end{Definition}

Given a  control policy $\Omega$, the induced random process $\left\{\boldsymbol{\chi}\left(t \right)\right\}$ is a controlled Markov chain with the following transition probability:
\txtblue{\begin{align}	
	  &\Pr\big[ \boldsymbol{\chi}\left(t+1 \right) \big| \boldsymbol{\chi}\left(t \right),  \Omega\big(\boldsymbol{\chi}\left(t \right) \big)\big] \label{trankernel} \\
	 =& \Pr \big[\mathbf{H}\left(t+1 \right) \big] \Pr \big[ \mathbf{Q} \left(t+1\right) \big| \boldsymbol{\chi}\left(t \right), \Omega\big(\boldsymbol{\chi}\left(t \right) \big) \big]	\notag \\
	 =& \Pr \big[\mathbf{H}\left(t+1 \right) \big] \prod_{k=1}^K \Pr \big[ Q_k \left(t+1\right) \big|Q_k(t), \mathbf{H}(t), \Omega\big(\boldsymbol{\chi}\left(t \right) \big) \big]	\notag
\end{align}
where $\Pr \big[ Q_k \left(t+1\right) \big|Q_k(t), \mathbf{H}(t), \Omega\big(\boldsymbol{\chi}\left(t \right) \big) \big]$ is the queue transition probability for the $k$-th Tx-Rx pair and is given by
\begin{align}	
&\Pr \big[ Q_k \left(t+1\right) \big|Q_k(t), \mathbf{H}(t), \Omega\big(\boldsymbol{\chi}\left(t \right) \big) \big] \notag 	\\
=&
 \left\{
	\begin{aligned}	 \label{trankern}
		& 1 \quad \text{if } Q_k(t+1) = \left[Q_k\left(t\right) - \mu_k  \tau \right]^+ \\
		&\hspace{2.8cm}+ R_k\left(\mathbf{H}\left( t\right), \mathbf{F}\left( t\right), \mathbf{U}_k (t) \right)  \tau   	\\
		& 0  \quad \text{otherwise}
	   \end{aligned}
   \right.
  \end{align}
Note that the last equality in (\ref{trankernel}) is due to the queue evolution equation in (\ref{Qdyn}). Hence, $\left\{Q_k(t) \right\}$ is a controlled Markov chain and the next transition $Q_k(t+1)$ only depends on $Q_k(t)$, $\mathbf{H}(t)$, and $\left(\mathbf{F},\mathbf{U}_k \right)$.}

Furthermore, we have the following definition on the admissible  control policy:
\begin{Definition}	[{Admissible Control Policy}]	\label{adddtdomain}
	{A policy $\Omega$ is  admissible if the following requirements are satisfied:}
	\begin{itemize}
		\item $\Omega$ is a unichain policy, i.e., the controlled Markov chain $\left\{\boldsymbol{\chi}\left(t \right)\right\}$ under $\Omega$ has a single recurrent class (and possibly some transient states) \cite{DP_Bertsekas}.
		\item {The queueing system under $\Omega$ is stable in the sense that   $\lim_{t \rightarrow \infty} \mathbb{E}^{\Omega} \big[ \sum_{k=1}^K Q_k(t)  \big] < \infty$, where $\mathbb{E}^{\Omega} $ means taking expectation w.r.t. the probability measure induced by the control policy $\Omega$. }~\hfill\IEEEQED
	\end{itemize}
\end{Definition}

\subsection{Multimedia Streaming Performance  and Cross-Layer  Problem Formulation}
The system performance of the multimedia system is characterized by the average transmit power of the BSs, playback interruption probability and buffer overflow probability of the mobile users.

Under an admissible control policy $\Omega$, the average power cost of BS $k$ starting from a given initial  state $\boldsymbol{\chi}\left(0\right)$ is given by
\begin{align}	\label{delay_cost}
	\overline{P}_k^\Omega\left(\left(\boldsymbol{\chi}\left(0 \right) \right)\right)  = \limsup_{T \rightarrow \infty} \frac{1}{T} \sum_{t=0}^{T-1} \mathbb{E}^{\Omega} \left[P_k\left( \mathbf{F}_k(t)\right) \right] 	
\end{align}	
where $P_k\left( \mathbf{F}_k\right)$ is defined in (\ref{defpow}).  Similarly, under an admissible control policy $\Omega$,  the playback interruption probability and buffer overflow probability of user $k$ are given by
\begin{align}
	  \overline{I}_k^{\Omega}\left(\left(\boldsymbol{\chi}\left(0 \right) \right)\right) =&  \limsup_{T \rightarrow \infty} \frac{1}{T} \sum_{t=0}^{T-1} \mathbb{E}^{\Omega} \left[1 \left(Q_k(t) < Q^l \right) \right] \notag \\
	\approx&   \limsup_{T \rightarrow \infty} \frac{1}{T} \sum_{t=0}^{T-1} \mathbb{E}^{\Omega} \left[e^{-\eta \left[Q_k(t) - Q^l \right]^+} \right]	\label{ind1}\\
	\overline{B}_k^{\Omega}\left(\left(\boldsymbol{\chi}\left(0 \right) \right)\right) =&  \limsup_{T \rightarrow \infty} \frac{1}{T} \sum_{t=0}^{T-1} \mathbb{E}^{\Omega} \left[1 \left(Q_k(t) > Q^h \right) \right] \notag \\
		\approx& \limsup_{T \rightarrow \infty} \frac{1}{T} \sum_{t=0}^{T-1} \mathbb{E}^{\Omega} \left[e^{-\eta \left[Q^h - Q_k(t) \right]^+}\right] 	\label{ind2}
\end{align}where $Q^l>0$ and $Q^h>0$ are  the target minimum  and maximum playback buffer size at the mobile users, respectively, and we require $Q^h>Q^l$. Fig.~\ref{fig1adsa} illustrates an example of the queue trajectory and the playback interruption/overflow events. During  the playback interruption, the multimedia playback is frozen, which  affects the end user experience. During the overflow event, the arrival packets are dropped and this causes wastage of the radio resource used to transmit the dropped packets.   For technicality, we use $e^{-\eta \left[Q_k - Q^l \right]^+}$ and $e^{-\eta \left[Q^h - Q_k \right]^+}$ as a smooth approximation for the indicator functions in (\ref{ind1}) and (\ref{ind2}), where $\eta >0$ is a parameter\footnote{The approximation is asymptotically accurate as $\eta \rightarrow \infty$.} of the smooth approximation.  \txtblue{Fig.~\ref{compareees1} and Fig.~\ref{compareees2} illustrate the comparison of the actual and approximate playback interruption and buffer overflow per-stage costs. It can be observed that the approximate per-stage playback interruption (or buffer overflow) cost    $e^{-\eta \left[Q_k - Q^l \right]^+}$ (or $e^{-\eta \left[Q^h - Q_k\right]^+}$) is very close to the  actual per-stage cost  $1 \left(Q_k < Q^l \right)$ (or $1 \left(Q_k> Q^h \right)$) for large values of $\eta$ (e.g., $\eta \geq 10$).}

\begin{figure}
\begin{minipage}[t]{0.45\textwidth}\centering
  \includegraphics[width=3.5in]{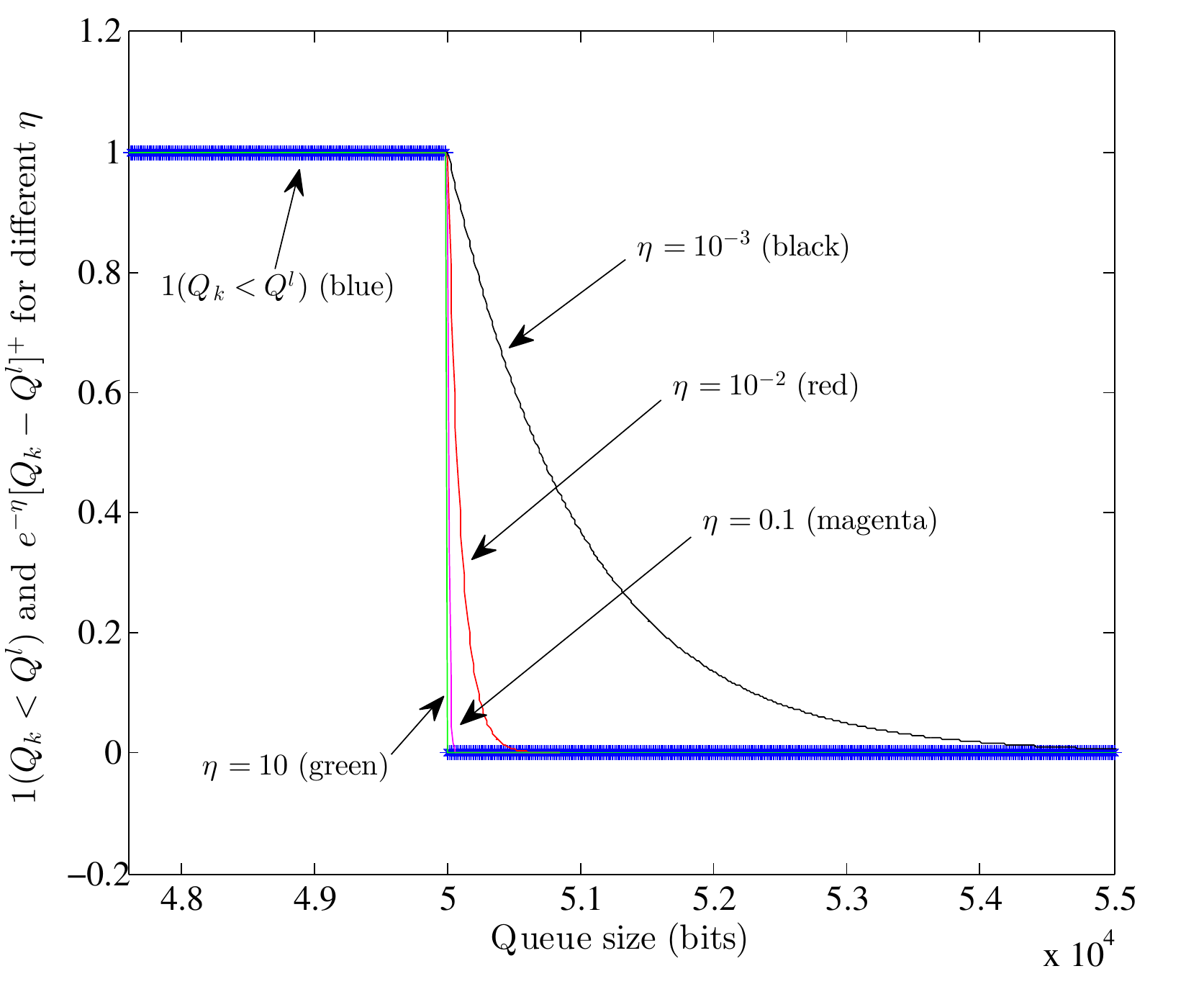}
  \caption{\txtblue{Actual per-stage playback interruption cost $1 (Q_k < Q^l )$ and the associated approximation $e^{-\eta \left[Q_k- Q^l \right]^+}$ for different values of $\eta$.}}
  \label{compareees1}
\end{minipage}
\hspace{0.05\textwidth}
\begin{minipage}[t]{0.45\textwidth}
  \centering	 \hspace{-0.4cm}
   \includegraphics[width=3.43in]{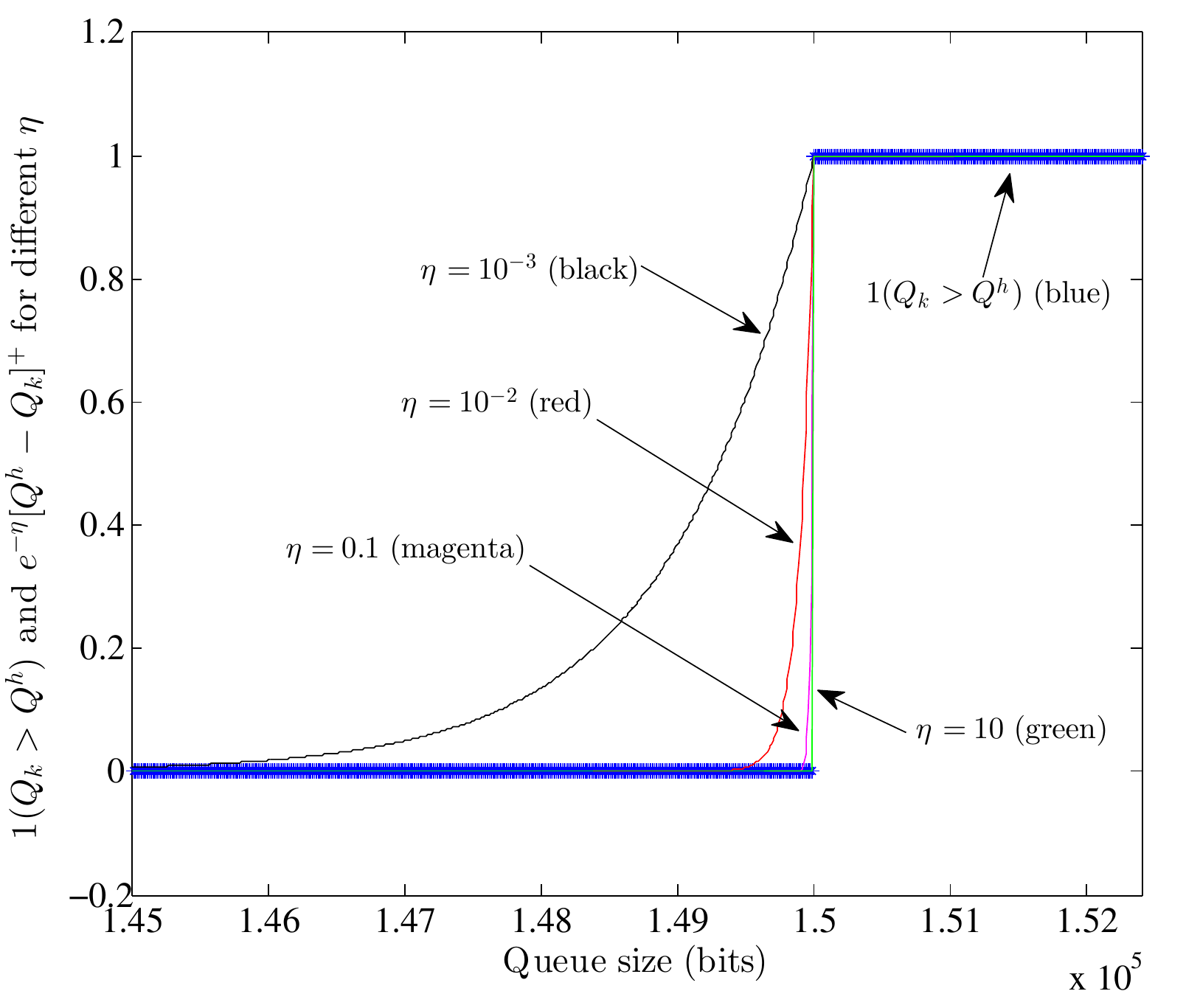}
 \caption{\txtblue{Actual  per-stage  buffer overflow cost $1 \left(Q_k > Q^h \right)$ and the associated approximation $e^{-\eta \left[Q^h - Q_k \right]^+}$ for different values of $\eta$.}}
  \label{compareees2}\end{minipage}	
\end{figure}

We consider a truly cross-layer framework for the MIMO precoder/decorrelator optimization with the optimization objective to be the weighted sum of the average transmit power, the average playback interruption and the average buffer overflow probabilities of the multimedia streaming applications. This is formally stated below. 
\begin{Problem}	 \label{IHAC_MDP}	\emph{(Stochastic Precoder and Decorrelator   Control Problem)}
	For some positive constants  $\boldsymbol{\gamma}=\left\{ \gamma_k>0: \forall k\right\}$ and $\boldsymbol{\beta}=\left\{ \beta_k>0: \forall k\right\}$, the stochastic precoder and decorrelator  control  problem is formulated as
\begin{align} 
	\min_{\Omega}	\quad  & L_{\boldsymbol{\gamma},\boldsymbol{\beta}}^{\Omega}\left( \boldsymbol{\chi}\left(0 \right)\right)		\label{perstagecost} \\
	 =&\sum_{k=1}^K \left( \overline{P}_k^\Omega\left(\left(\boldsymbol{\chi}\left(0 \right) \right)\right)   +  \gamma_k \overline{I}_k^{\Omega}\left(\left(\boldsymbol{\chi}\left(0 \right) \right)\right) +\beta_k  \overline{B}_k^{\Omega}\left(\left(\boldsymbol{\chi}\left(0 \right) \right)\right)  \right) 	\notag \\
	 =&  \limsup_{T \rightarrow \infty} \frac{1}{T} \sum_{t=0}^{T-1} \mathbb{E}^{\Omega}  \left[c\left(\mathbf{Q} \left(t\right), \Omega\left(\boldsymbol{\chi}\left(t\right) \right)\right)    \right]\notag 	
\end{align}
where $\boldsymbol{\gamma}$ and $\boldsymbol{\beta}$ measures the relative importances\footnote{$\boldsymbol{\gamma}$ and $\boldsymbol{\beta}$  can also be interpreted as the corresponding Lagrange Multipliers  associated with the playback interruption probabilities and buffer overflow probabilities of the $K$ users \cite{mdpapp2}.} and tradeoffs of the playback interruption probability and buffer overflow probability. $c\left(\mathbf{Q}, \mathbf{F}\right)=\sum_{k=1}^K c_k\left(Q_k,\mathbf{F}_k\right)$ is the per-stage cost function with $c_k\left(Q_k,\mathbf{F}_k\right)=\mathrm{Tr}\left(  \mathbf{F}_k  \mathbf{F}_k^\dagger\right)+\gamma_k e^{-\eta \left[Q_k - Q^l \right]^+}+\beta_k e^{-\eta \left[Q^h - Q_k \right]^+} $.~\hfill\IEEEQED
\end{Problem}

Note that the two technical conditions in Definition \ref{adddtdomain} on the admissible policy  ensure that there is a unique solution to Problem \ref{IHAC_MDP}. Furthermore,  Problem \ref{IHAC_MDP} is an infinite horizon average cost  MDP,  which is well-known to  be a very difficult problem \cite{surveydelay}. In the next  subsection, by exploiting the special structure in our problem, we derive an equivalent optimality equation to simplify the   MDP problem.

\subsection{Optimality Conditions and Approximate Optimality Equation}
While the MDP in Problem \ref{IHAC_MDP} is  difficult in general, we utilize  the i.i.d. assumption of the CSI to derive an \emph{equivalent optimality equation} as summarized below.

\begin{Theorem} [Sufficient Conditions for Optimality]	\label{LemBel}
	For any given  $\boldsymbol{\gamma}$ and $\boldsymbol{\beta}$, assume there exists a ($\theta^\ast, \{ V^\ast\left(\mathbf{Q}  \right) \}$) that solves the following \emph{equivalent optimality equation}:
	\begin{align}	
		 &\theta^\ast \tau + V^\ast \left(\mathbf{Q}  \right), \hspace{2cm} \forall \mathbf{Q}  \in \boldsymbol{\mathcal{Q}} \label{OrgBel}	\\
		 = & \mathbb{E} \bigg[\min_{\mathbf{F}, \mathbf{U}}\Big[ c\left(\mathbf{Q}, \mathbf{F}\right) \tau +  \sum_{\mathbf{Q} '}\Pr \big[ \mathbf{Q} '\big| \boldsymbol{\chi}, \mathbf{F}, \mathbf{U} \big]V^\ast \left(\mathbf{Q}  '\right) \Big]   \bigg| \mathbf{Q}  \bigg]\notag
	\end{align} Furthermore,  for all  admissible control policy $\Omega$ and initial queue state $\mathbf{Q} \left(0 \right)$, $V^\ast$ satisfies the following \emph{transversality condition}:
	\begin{align}	\label{transodts}
	\lim_{T \rightarrow \infty} \frac{1}{T}\mathbb{E}^{\Omega}\left[ V^\ast\left(\mathbf{Q} \left(T \right) \right) |\mathbf{Q} \left(0 \right)\right]=0
\end{align}
	Then, $\theta^\ast=\underset{\Omega}{\min} L_{\boldsymbol{\gamma},\boldsymbol{\beta}}^{\Omega}\left( \boldsymbol{\chi}\left(0 \right)\right) $ is the optimal average cost for any initial state  $\boldsymbol{\chi}\left(0 \right) $ and $V^\ast\left(\mathbf{Q} \right)$ is called the \emph{value function}. If $\left(\mathbf{F}^\ast, \mathbf{U}^\ast \right)$ attains the minimum of the R.H.S. in (\ref{OrgBel}) for given $\boldsymbol{\chi}$, then the optimal control policy of Problem \ref{IHAC_MDP} is given by $\Omega^*\left(\boldsymbol{\chi} \right) = \left(\mathbf{F}^\ast, \mathbf{U}^\ast \right)$.~\hfill\IEEEQED	
\end{Theorem}

\begin{proof}
	please refer to Appendix A.
\end{proof}
Solving (\ref{OrgBel}) is a difficult problem because it corresponds to solving a series of fixed point equations w.r.t. $(\theta^\ast,\{ V^\ast(\mathbf{Q}) \} )$ which involves exponentially many equations and  unknowns. This explains why standard solutions such as value iteration and policy iteration \cite{Cao}, \cite{DP_Bertsekas} have exponential complexity w.r.t. $K$. Instead of solving (\ref{OrgBel}) directly, we exploit the timescale separation property between the slot duration $\tau$ and the interruption/overflow events and establish an approximate optimality equation to further simplify our problem.
\begin{Corollary}	[Approximate Optimality Equation]	\label{cor1}
For any given  $\boldsymbol{\gamma}$ and $\boldsymbol{\beta}$, if
\begin{itemize}
\item	there is a unique  ($\theta^\ast, \{ V^\ast\left(\mathbf{Q}  \right) \}$) that satisfies the optimality equation and transversality condition in Theorem \ref{LemBel}.

\item	there exist $\theta$ and $V\left( \mathbf{Q} \right)$ of class\footnote{$f(\mathbf{x})$ ($\mathbf{x}$ is a $K$-dimensional vector) is of class $\mathcal{C}^2(\mathbb{R}_+^K)$, if the first and second order partial derivatives of $f(\mathbf{x})$ w.r.t. each element of $\mathbf{x}$ are continuous when $\mathbf{x}\in \mathbb{R}_+^K$.} $\mathcal{C}^2(\mathbb{R}_+^K)$ that solve the following \emph{approximate optimality equation}:
\begin{align}
	&\theta = \mathbb{E}\bigg[ \min_{\mathbf{F}, \mathbf{U} }  \Big[c\left(\mathbf{Q}, \mathbf{F}\right) + 	\hspace{2cm} \forall \mathbf{Q}  \in \boldsymbol{\mathcal{Q}} \\ \label{bellman2}
	 & \hspace{1cm} \sum_{k=1}^K \frac{\partial V \left(\mathbf{Q}  \right) }{\partial Q_k} \big( R_k\left(\mathbf{H}, \mathbf{F}, \mathbf{U}_k\right) -\mu_k \big) \Big]\bigg| \mathbf{Q}   \bigg]\notag
\end{align}Furthermore, for all admissible  control policy $\Omega$ and initial queue state $\mathbf{Q} \left(0 \right)$, the transversality condition  in (\ref{transodts}) is satisfied for $V$.
\end{itemize}
then, we have
\begin{align}
	\theta^\ast=\theta+o(1), \quad V^\ast\left(\mathbf{Q}  \right)=V\left(\mathbf{Q}  \right)+o(1), \quad \forall \mathbf{Q}  \in \boldsymbol{\mathcal{Q}}
\end{align}
where the error term $o(1)$ asymptotically goes to zero  for sufficiently small slot duration $\tau$.~\hfill\IEEEQED
\end{Corollary}
\begin{proof}
Please refer to Appendix B.
\end{proof}

Corollary \ref{cor1} states that the difference between ($\theta, \{ {V}\left(\mathbf{Q}  \right) \}$) obtained in (\ref{bellman2}) and ($\theta^\ast, \{ {V^\ast}\left(\mathbf{Q}  \right) \}$) in (\ref{OrgBel}) is asymptotically small w.r.t. the slot duration $\tau$. Therefore, we can focus on solving the approximate optimality equation in  (\ref{bellman2}), which is a simpler problem than solving the original optimality equation in (\ref{OrgBel}).

\section{Closed-Form Approximate Value Function based on Calculus Approach}	\label{weakintder}
In this section, we  adopt a calculus  approach to obtain a closed-form approximation of the  value function. Specifically, we shall exploit the  weak interference property  and  utilize the perturbation theory to obtain the approximate value function.

\subsection{Multi-dimensional PDE}
We first have the following theorem for solving the approximate optimality equation in (\ref{bellman2}):
\begin{Theorem}	\label{HJB1}	 \emph{(Calculus Approach on Solving the Approximate Optimality  Equation)}	
	Assume there exist $c^\infty$ and  $J\left( \mathbf{Q} ;L\right)$ of class $\mathcal{C}^2(\mathbb{R}_+^K)$   that satisfy
	\begin{itemize}
		\item the following multi-dimensional PDE:
		\begin{align}	
		&\hspace{-0.3cm} \mathbb{E}\bigg[  \min_{\mathbf{F}, \mathbf{U}} \bigg[ c\left(\mathbf{Q}, \mathbf{F}\right) + \hspace{3cm}   \mathbf{Q}  \in  \mathbb{R}_+^K  \label{cenHJB}  \\
		&  \sum_{k=1}^K  \frac{\partial  J \left(\mathbf{Q} ;{L}\right)} {\partial Q_k} \left(R_k\left(\mathbf{H}, \mathbf{F}, \mathbf{U}_k \right)    -\mu_k\right) \bigg]\bigg| \mathbf{Q}  \bigg] - c^{\infty} =0	\notag 
\end{align}with boundary condition $J\left(Q_1^\star,\dots, Q_K^\star ; L\right)=0$, where $Q_k^\star$ ($\forall k$) are some given constant queue values.
		\item $\frac{\partial  J \left(\mathbf{Q} ;{L}\right)} {\partial Q_k}>0$ for sufficiently large $Q_k$ for all $k$.
		\item $J\left(\mathbf{Q} ; L \right)=\mathcal{O}\left(\sum_{k=1}^K Q_k \right)$.
	\end{itemize}
	Then, we have
		\begin{align}	
		\theta^\ast= c^{\infty}+o(1), \quad V^\ast\left(\mathbf{Q}  \right)=J \left(\mathbf{Q} ;{L}\right)+o(1), \quad \forall \mathbf{Q}  \in \boldsymbol{\mathcal{Q}}	\label{15resu}
	\end{align}~\hfill\IEEEQED
\end{Theorem}
\begin{proof}
	please refer to Appendix C.
\end{proof}

As a result, solving the approximate optimality equation in (\ref{bellman2}) is transformed into a calculus problem of solving the  PDE in (\ref{cenHJB}). However, the PDE  is still a $K$-dimensional non-linear PDE, which is in general very challenging. To obtain a closed-form approximation of $V^\ast(\mathbf{Q} )$, we apply perturbation analysis to a base PDE as shown in the next subsection.

\subsection{Perturbation Approximation  of $J\left(\mathbf{Q} ; {L}\right) $}
The solution of the multi-dimensions PDE in (\ref{cenHJB}) depends on the worst-case cross channel path gain $L$ and hence, the $K$-dimensional PDE  can be regarded as a perturbation of a base PDE defined below.
\begin{Definition} [Base PDE]	\label{base_sys}
	A  base PDE is the PDE in (\ref{cenHJB}) with $L=0$.~\hfill\IEEEQED
\end{Definition}

We then  study  the  base PDE  and use  $J\left(\mathbf{Q} ; {0}\right) $ to obtain a closed-form approximation of $J \left(\mathbf{Q} ;{L}\right)$. We have the following lemma summarizing the decomposable structure of  $c^\infty$ and $J\left(\mathbf{Q} ; {0}\right) $:

\begin{Lemma}	[Decomposable Structure of  $c^\infty$ and $J\left(\mathbf{Q} ; {0}\right)$]	\label{linearAp}
	If\footnote{These conditions on the weights $\boldsymbol{\gamma}$ and $\boldsymbol{\beta}$ are imposed to make sure there is a solution for  (\ref{cenHJB}) in Theorem \ref{HJB1}. Qualitatively, if $\beta_k$ is too small, the power cost in Problem \ref{IHAC_MDP} will dominate and the user $k$ will not be served at all (they will be allocated zero power). These conditions are used to avoid such uninteresting degenerated case.} $e^{\eta \left(Q^l-Q^h  \right)} < \frac{\gamma_k}{\beta_k} < e^{\eta \left(Q^h-Q^l  \right)}$ and $\beta_k>c_k^\infty$,  then $c^\infty$ and $J \left( \mathbf{Q} ; {0} \right)$ in the  base PDE  has the following decomposable structure:
	\begin{align}   \label{linearA}
		c^\infty=\sum_{k=1}^K c_k^\infty, \quad J \left( \mathbf{Q} ;{0} \right) = \sum_{k=1}^K J_k \left(Q_k \right)
	\end{align}
	where $c_k^\infty$ is given by  (\ref{cinfequ})  and  $J_k' \left(Q_k \right)$ is determined\footnote{The optimal control policy by solving  (\ref{cenHJB}) only requires the partial derivatives of the value functions $\{\frac{\partial J(\mathbf{Q};L)}{\partial Q_k}:\forall k\}$, so we can focus on deriving $J_k'(Q_k)$ in the based PDE.} by the fixed point equation in (\ref{fixedequ}) in Appendix D. Furthermore, we have the following asymptotic property of $J_k \left( Q_k \right)$:
	\begin{align}	\label{asympoticJ_k}
		J_k \left( Q_k \right) = C_k Q_k, \quad \text{as } Q_k \rightarrow \infty
	\end{align}
	where $C_k=\frac{\beta_k - c_k^\infty}{\mu_k}$ is a positive constant. ~\hfill\IEEEQED
\end{Lemma}
\begin{proof}
Please refer to Appendix D.
\end{proof}

Next, we  approximate  $J\left(\mathbf{Q} ; L\right)$ as a perturbation of $J\left(\mathbf{Q} ; 0\right)$. Using perturbation analysis, we establish the following theorem on the approximation of $J\left(\mathbf{Q} ; L\right)$.
\begin{Theorem}	[Perturbation Approximation of $J\left(\mathbf{Q} ; L\right)$]	\label{ErrorEg2}	
	 $J\left(\mathbf{Q} ; L\right)$ can be approximated by $J\left(\mathbf{Q} ; 0\right)$ and the first order perturbation term is given by
	 \begin{align}\label{Ex1Error}
	 	J\left(\mathbf{Q} ; L\right) = \sum_{k=1}^K J_k(Q_k) -\sum_{k=1}^K \sum_{j \neq k}  L_{kj} h_{kj}\left(Q_k,Q_j \right)+ \mathcal{O}\left(L^2 \right)
	 \end{align}
	 where $h_{kj}\left(Q_k,Q_j \right)=o(1)$, if  either $Q_k>Q_k^\star$ or $Q_j>Q_j^\star$, and $h_{kj}\left(Q_k,Q_j \right)=E_{kj} \left( Q_k-Q_k^\star\right) +E_{jk} ( Q_j-Q_j^\star)+o(Q_k)+o(Q_j)$ ($k \neq j$), otherwise.    $Q_k^\star = \frac{Q^l+Q^h}{2} + \frac{1}{2 \eta} \ln \frac{\gamma_k}{\beta_k} \in (Q^l, Q^h )$ and $E_{kj}=\frac{\ln 2 \left( c_k^1 D_k +c_k^2 \right)\left(c_j^1D_j+c_j^2 \right) }{2 d W (  \mu_k-c_k^1 \ln(-D_k)-c_k^3)} $ is a  constant ($c_k^1$, $c_k^2$, $c_k^3$ and $D_k$ are given in (\ref{111equ})--(\ref{fixedpointws}) in Appendix E.)~\hfill\IEEEQED
\end{Theorem}
\begin{proof}
Please refer to Appendix E.		
\end{proof}

Finally, based on Theorem \ref{HJB1} and Theorem \ref{ErrorEg2}, we propose the following closed-form approximation of the relative value function:
\begin{align}	\label{finalapprox1}
	V^\ast\left( \mathbf{Q}  \right) \approx  \widetilde{V}\left( \mathbf{Q}  \right) \triangleq  \sum_{k=1}^K J_k(Q_k) -\sum_{k=1}^K \sum_{j \neq k}  L_{kj} \widetilde{h}_{kj}\left(Q_k,Q_j \right)
\end{align}
 where $\widetilde{h}_{kj}\left(Q_k,Q_j \right)=0$, if  either $Q_k>Q_k^\star$ or $Q_j>Q_j^\star$, and $\widetilde{h}_{kj}\left(Q_k,Q_j \right)=E_{kj} ( Q_k-Q_k^\star) +E_{jk} ( Q_j-Q_j^\star)$, otherwise.

\begin{figure}
\centering
  \includegraphics[width=3.5in]{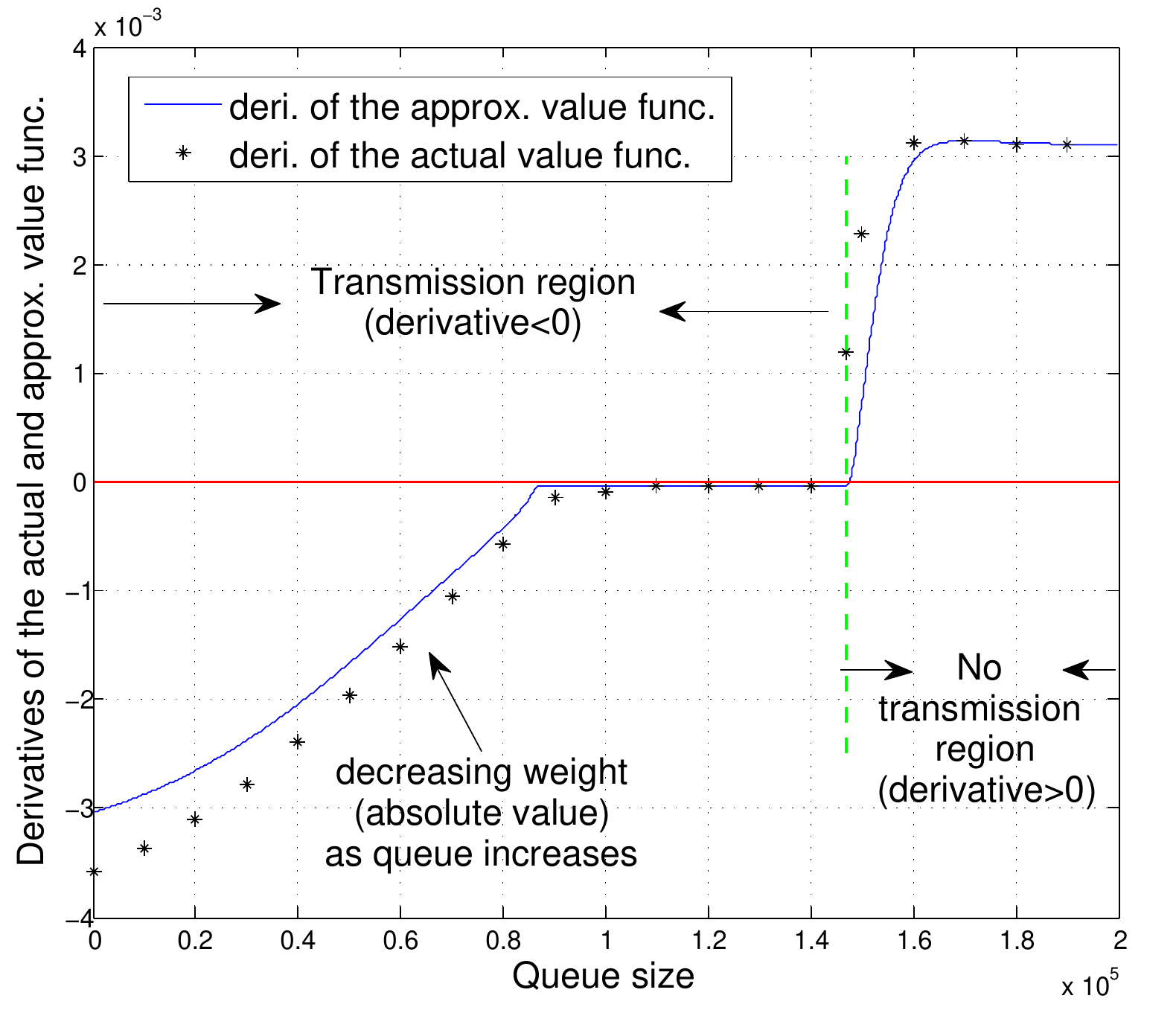}
  \caption{Derivatives of the actual value function $\frac{\partial {V^\ast}\left(\mathbf{Q} \right)}{\partial Q_k}$ and approximate value function $\frac{\partial \widetilde{V}\left(\mathbf{Q} \right)}{\partial Q_k}$ versus the playback queue $Q_k$ with $\mathbf{Q}=(Q_1, 50K, 100K, 150K)$ and $\eta=50$. The system parameters are configured as in the simulations in Section \ref{xxsim}.}
  \label{fig1wdas}
\end{figure}

Furthermore, based on Lemma \ref{linearAp}  and (\ref{finalapprox1}), we have  $\frac{\partial  J \left(\mathbf{Q} ;{L}\right)} {\partial Q_k}>0$ for sufficiently large $Q_k$ ($\forall k$),  and  $\widetilde {V}\left(\mathbf{Q}  \right) =\mathcal{O}\left(\sum_{k=1}^K Q_k \right)$.  Based on Theorem \ref{HJB1} and Theorem 3, the approximation error between the optimal value function $V^\ast \left(\mathbf{Q}  \right)$ in Theorem \ref{LemBel} and the closed-form approximate value function  $\widetilde {V}\left(\mathbf{Q}  \right)$ in (\ref{finalapprox1})  is $\mathcal{O}(L)+ o(1)$. In other words, the error terms are asymptotically small w.r.t. the worst-case cross channel path gain\footnote{\txtblue{Based on the Wi-Fi CSMA/CA example in Section \ref{frisss}, $L$ is very small and hence, the first order approximation in (\ref{finalapprox1}) is quite accurate.}} and the slot duration. Fig.~\ref{fig1wdas} illustrates the comparison of the derivatives of the  actual (optimal) value function and the approximate value function. In the next section, we derive a low complexity  control policy using the closed-form approximate value function in (\ref{finalapprox1}).

\section{low complexity MIMO precoders/decorrelators  based on WMMSE}	\label{algorithm_sec}
In this section, we use the closed-form approximate value function  in (\ref{finalapprox1}) to capture the urgency information of the data flows and to  obtain low complexity dynamic precoder and decorrelator   control.  We first show that minimizing the R.H.S. of the Bellman equation using the closed-form approximate value function is equivalent to solving a collection of  per-stage control problems. We further transform the per-stage problem into an equivalent weighted sum-MSE minimization problem and obtain a low complexity dynamic precoder and decorrelator  control  algorithm by solving this problem.

\subsection{Equivalent Per-Stage Control Problem}
Using the approximate value function in (\ref{finalapprox1}) and Corollary \ref{cor1}, the per-stage control problem (for each state realization $\boldsymbol{\chi}$) is given by\footnote{From  (\ref{finalapprox1}), $\frac{\partial \widetilde {V}\left(\mathbf{Q} \right)}{\partial Q_k} = J_k'(Q_k)-\sum_{j\neq k}\big(L_{kj}\widetilde{h}_{kj}'(Q_k,Q_j)+L_{jk}\widetilde{h}_{jk}'(Q_k,Q_j)\big)$, where $\widetilde{h}_{kj}'(Q_k,Q_j)=0$, if either $Q_k>Q_k^\star$ or $Q_j>Q_j^\star$, and $\widetilde{h}_{kj}'\left(Q_k,Q_j \right)=E_{kj}$, otherwise.}
\begin{align}	\label{utility1}
		\min _{\mathbf{F},\mathbf{U} }  \ \sum_{k=1}^K  \Big(\mathrm{Tr}\left(  \mathbf{F}_k  \mathbf{F}_k^\dagger\right)  +    \frac{\partial \widetilde{V}\left(\mathbf{Q} \right)}{\partial Q_k}  R_k\left(\mathbf{H}, \mathbf{F}, \mathbf{U}_k\right)  \Big)
	\end{align}
Note that the weights $\left\{ \frac{\partial \widetilde{V}\left(\mathbf{Q} \right)}{\partial Q_k}  : \forall k\right\}$ in the above per-stage control problem is determined by the instantaneous QSI. Fig.~\ref{fig1wdas}  illustrates $\frac{\partial \widetilde{V}\left(\mathbf{Q} \right)}{\partial Q_k}$ versus $Q_k$. When $Q_k$ is large, $-\frac{\partial \widetilde{V}\left(\mathbf{Q} \right)}{\partial Q_k}$ is small and hence, the priority of the $k$-th flow is reduced. This is reasonable because when $Q_k$ in the playback buffer of user $k$ is large, the $k$-th flow can withstand intermittent fading or reduction in the instantaneous arrivals for some time before playback interruption occurs.

For given precoding matrices $\mathbf{F}$ and state realization $\boldsymbol{\chi}$, the optimal decoding matrices of the mobile users $\mathbf{U}^\ast\left( \mathbf{F}\right)=\left\{\mathbf{U}_k^\ast\left( \mathbf{F}\right): \forall k \right\}$ are  given by the MMSE receiver \cite{WMMSE}:
\begin{align}	\label{mmserec}
	\mathbf{U}_k^\ast \left( \mathbf{F}\right)=\mathbf{J}_{k}^{-1}\left( \mathbf{F}\right) L_{kk} \mathbf{H}_{kk}\mathbf{F}_k, \quad \forall k
\end{align}
where $\mathbf{J}_{k}\left( \mathbf{F}\right)=\sum_{j=1}^K L_{kj}\mathbf{H}_{kj}\mathbf{F}_j \mathbf{F}_j^\dagger \mathbf{H}_{kj}^\dagger + \mathbf{I}$ is the downlink  signal plus noise covariance matrix.  Using the MMSE receiver in (\ref{mmserec}), the per-stage control problem in (\ref{utility1}) can be further transformed into the following equivalent form:
\begin{align}	\label{utility}
		\min _{\mathbf{F} }  \ \sum_{k=1}^K  \Big(\mathrm{Tr}\left(  \mathbf{F}_k  \mathbf{F}_k^\dagger\right)  +    \frac{\partial \widetilde{V}\left(\mathbf{Q} \right)}{\partial Q_k}  R_k\left(\mathbf{H}, \mathbf{F}, \mathbf{U}_k^\ast\left( \mathbf{F}\right)\right)  \Big)
	\end{align}
where $R_k\left(\mathbf{H}, \mathbf{F}, \mathbf{U}_k^\ast\left( \mathbf{F}\right)\right) = W \log_2 \det \big( \mathbf{I}+ L_{kk} \mathbf{H}_{kk}\mathbf{F}_k \mathbf{F}_k^\dagger \mathbf{H}_{kk}^\dagger \big(\sum_{j \neq k} L_{kj}\mathbf{H}_{kj}\mathbf{F}_j \mathbf{F}_j^\dagger \mathbf{H}_{kj}^\dagger+\mathbf{I} \big)^{-1}\big)$.
\begin{Remark}	[Interpretation of (\ref{utility})]
	The precoding matrices obtained by solving (\ref{utility}) is adaptive to both the CSI and the QSI. Furthermore, for sufficiently  large $Q_k$, $\frac{\partial \widetilde{V}\left(\mathbf{Q} \right)}{\partial Q_k}$ is positive, which results in the associated optimal precoding matrices\footnote{Please refer to Lemma \ref{lemmaappr} in Appendix C for the detailed proof.} to be $\mathbf{0}$.  Therefore, for given QSI realization $\mathbf{Q}$, we can focus on the MIMO precoder/decorrelator design for the set of users $\mathcal{I}_k^\mathbf{Q}=\left\{k: \frac{\partial \widetilde{V}\left(\mathbf{Q} \right)}{\partial Q_k} < 0 \right\}$, while the precoders/decorrelators  for the other users  ($k \notin \mathcal{I}_k^\mathbf{Q}$) are set to be $\mathbf{0}$.~\hfill\IEEEQED
\end{Remark}

\subsection{Low Complexity MIMO Precoders/Decorrelators Solution}
The per-stage problem in (\ref{utility}) can be further transformed into the following  weighted sum-MSE minimization problem  \cite{WMMSE}:
\begin{Problem}	[Weighted Sum-MSE  Minimization Problem]	\label{mmseprob}
\begin{align}
	\min _{\mathbf{F}, \mathbf{Z}, \mathbf{K}}  \ \sum_{k \in \mathcal{I}_k^\mathbf{Q}}  \big( \mathrm{Tr}(  \mathbf{F}_k  \mathbf{F}_k^\dagger)  -   \frac{\partial \widetilde{V}\left(\mathbf{Q} \right)}{\partial Q_k}  \frac{W}{\ln 2} \notag \\
	\big( \mathrm{Tr}\left(\mathbf{Z}_k \mathbf{E}_k\left( \mathbf{F}, \mathbf{K}\right) \right)- \ln \mathrm{det} \mathbf{Z}_k \big)\big)\label{utility1}
\end{align} where we denote $\mathbf{Z}=\left\{\mathbf{Z}_k:  k\in\mathcal{I}_k^\mathbf{Q} \right\}$ and $\mathbf{Z}_k\succeq0$ is a  weight for user $k$,   $\mathbf{K}=\left\{\mathbf{K}_k:  k\in\mathcal{I}_k^\mathbf{Q}  \right\}$  and $\mathbf{E}_k$ is the MSE  given by
\begin{align}		
	&\mathbf{E}_k\left( \mathbf{F}, \mathbf{K}\right) =\left(\mathbf{I}-\sqrt{L_{kk}}\mathbf{K}_k^\dagger \mathbf{H}_{kk} \mathbf{F}_k \right) \left(\mathbf{I}-\sqrt{L_{kk}}\mathbf{K}_k^\dagger \mathbf{H}_{kk} \mathbf{F}_k \right)^\dagger	\notag \\
	&+ \sum_{j \neq k \atop j \in\mathcal{I}_k^\mathbf{Q}}L_{kj}\mathbf{K}_k^\dagger \mathbf{H}_{kj} \mathbf{F}_j \mathbf{F}_j^\dagger  \mathbf{H}_{kj}^\dagger \mathbf{K}_k + \mathbf{K}_k^\dagger \mathbf{K}_k \label{mmesdef}
\end{align}~\hfill\IEEEQED
\end{Problem}
\begin{Lemma}	\label{equivalenceprobs}	\emph{(Relationship Between the Problems in (\ref{utility}) and (\ref{utility1}))}	
	The problem in (\ref{utility1}) is equivalent to the problem in (\ref{utility}), i.e., the global optimal solution $\mathbf{F}^{\ast}$ for the two problems are identical.~\hfill\IEEEQED
\end{Lemma}
\begin{proof}
The proof follows similar approach as in [17, Theorem 1]. Details are omitted due to page limit.  		
\end{proof}

The sum-MSE  cost function in (\ref{utility1}) is not jointly convex in all  the optimization variables $\left\{\mathbf{F}_k, \mathbf{Z}_k, \mathbf{K}_k: \right.\\ \left. k \in \mathcal{I}_k^\mathbf{Q}  \right\}$, but it is convex in each of $\left\{\mathbf{F}_k, \mathbf{Z}_k, \mathbf{K}_k:  k\in \mathcal{I}_k^\mathbf{Q} \right\}$ while holding the others fixed. Therefore, we propose to use an alternating iterative  algorithm to solve the WMMSE problem in Problem \ref{mmseprob}. In particular, we minimize the sum-MSE cost function by sequentially updating one  of $\left\{\mathbf{F}_k, \mathbf{Z}_k, \mathbf{K}_k:  k \in \mathcal{I}_k^\mathbf{Q} \right\}$ and fixing the others. The precoder and decorrelator  control algorithm based on WMMSE is given as follows:

\begin{Algorithm}	\label{wmmsealg} \emph{(Low Complexity Dynamic Precoder and Decorrelator  Control:)}
	\begin{itemize}		
		\item	\textbf{Step 1 [Initialization]:} Set $n=0$ and each BS $k$ initializes $\mathbf{F}_k(0)$.
		\item	\textbf{Step 2 [Message Passing between BSs and Mobile Users]:} Each user $k$ broadcasts its local QSI to the  $K$ BSs, and then each BS $k$ calculates $\left\{\frac{\partial \widetilde{V}\left(\mathbf{Q} \right)}{\partial Q_k}:\forall k\right\}$ locally. If $\frac{\partial \widetilde{V}\left(\mathbf{Q} \right)}{\partial Q_k}< 0$, the $k$-th Tx-Rx pair will participate in the precoder/decorrelator iterative calculations in the current slot. Otherwise, the associated precoder/decorrelator are set to be $\mathbf{0}$.  
		\item	\textbf{Step 3 [Update on $\mathbf{K}$ and $\mathbf{Z}$]:} Each BS $k$ ($k \in \mathcal{I}_k^\mathbf{Q}$) informs the associated mobile user $k$ of the updated $\mathbf{F}_k$. Each user $k$ ($k \in \mathcal{I}_k^\mathbf{Q}$) \txtblue{locally estimates the  downlink signal plus noise covariance matrix $\mathbf{J}_{k}\left( \mathbf{F}\right)$,} and updates $\mathbf{K}_k$ and $\mathbf{Z}_k$ according to the following equations:
			\begin{align}
				 & \mathbf{K}_k\left( n + 1 \right) = \mathbf{J}_{k}^{-1}\left( \mathbf{F}(n)\right) \sqrt{L_{kk}}\mathbf{H}_{kk}\mathbf{F}_k\left( n \right)	 \label{inv1}\\
				 & \mathbf{Z}_k\left( n + 1 \right) = \left(\mathbf{I}-  \mathbf{K}_k^\dagger \left( n \right)\sqrt{L_{kk}} \mathbf{H}_{kk} \mathbf{F}_k\left( n \right) \right)^{-1}		\label{inv2}
			\end{align}
		\item	\textbf{Step 4 [Update on $\mathbf{F}$]:} Each user $k$ ($k \in \mathcal{I}_k^\mathbf{Q}$) feeds back the updated $\mathbf{K}_k, \mathbf{Z}_k$   to each associated  BS $k$. Each BS $k$ locally updates $\mathbf{F}_k$  according to the following equations:
			\begin{align}
				& {\mathbf{F}}_k\left( n + 1 \right)=-\frac{\partial \widetilde{V}\left(\mathbf{Q} \right)}{\partial Q_k}\frac{W}{\ln 2} \Big(\sum_{j}-\frac{\partial \widetilde{V}\left(\mathbf{Q} \right)}{\partial Q_j} \frac{W}{\ln 2} L_{jk} 		\notag  \\
				&\mathbf{H}_{jk}^\dagger  \mathbf{K}_j \left( n+1 \right)  \mathbf{Z}_j\left( n+1 \right) \mathbf{K}_j^\dagger \left( n+1 \right)  \mathbf{H}_{jk} +  \mathbf{I} \Big)^{-1} \notag \\
				&\sqrt{L_{kk}} \mathbf{H}_{kk}^\dagger \mathbf{K}_k\left( n+1 \right) \mathbf{Z}_k\left( n+1 \right)\label{inv3} 	
			\end{align}
			\item	\textbf{Step 5 [Termination]:}  Set $n = n + 1$ and go to Step 3 until a certain termination condition is satisfied.~\hfill\IEEEQED
	\end{itemize}
\end{Algorithm}

\begin{figure}
\subfigure[\txtblue{Signaling flow with explicit QSI feedback.}]{
\hspace{1.95cm}\includegraphics[width=2.75in]{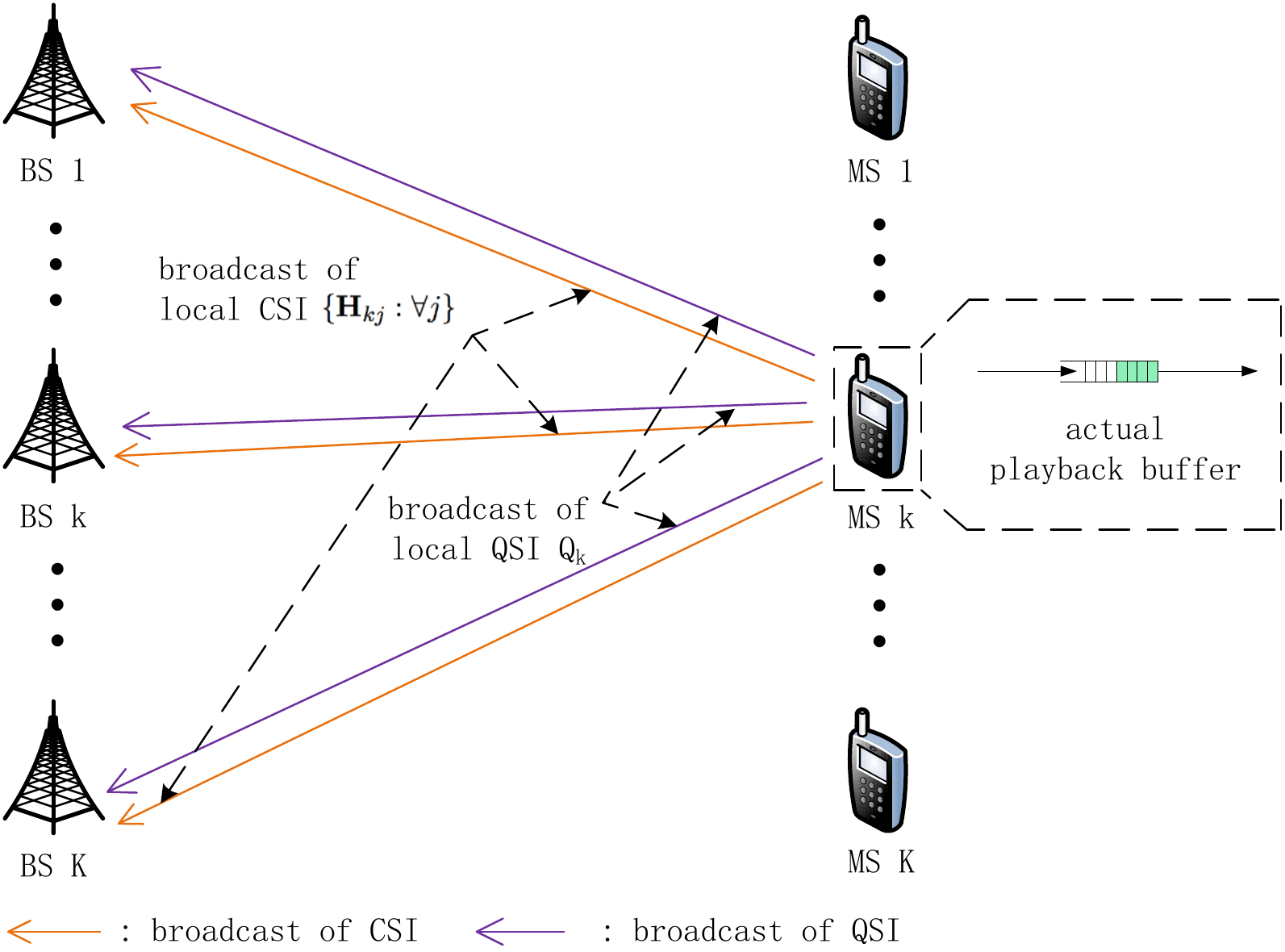}\label{fig1wdassadaasda}}
\hspace{0.5cm}
\subfigure[\txtblue{Signaling flow with virtual queue  at each BS and without explicit QSI feedback.}]{
\includegraphics[width=3.5in]{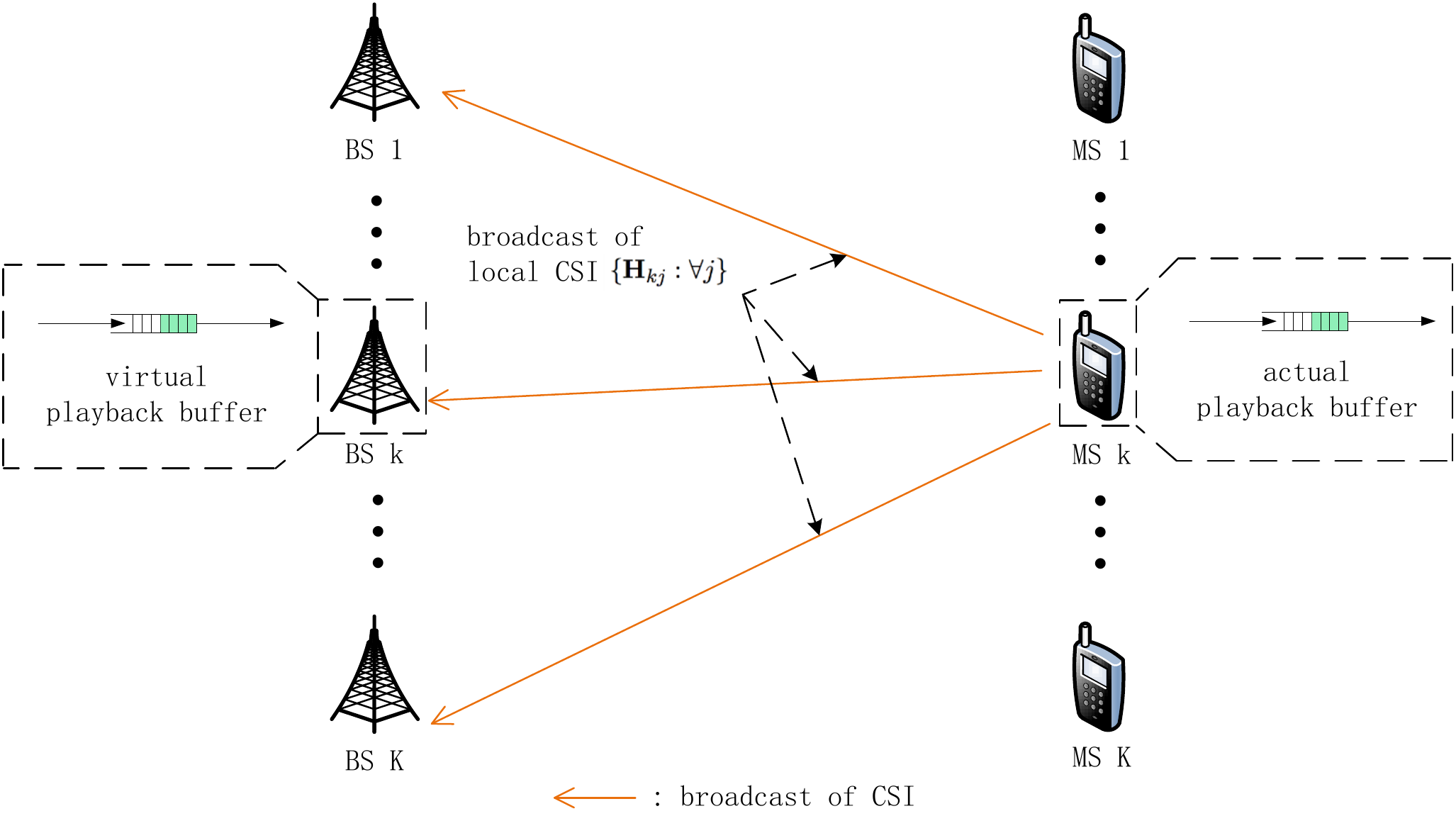}\label{fig1wdassadab}}
\centering \caption{\txtblue{Illustrations of the signaling flow of Algorithm 1 and  the virtual queue  at the BS for each Tx-Rx pair $k$.}}
\label{fig1wdassada}
\end{figure}

\begin{Lemma}	[Convergence Property of Algorithm \ref{wmmsealg}]	\label{globalopt}
	Any limiting point $\left\{\mathbf{F}(\infty), \mathbf{Z}(\infty), \mathbf{K}(\infty) \right\}$ of Algorithm \ref{wmmsealg} is a   stationary point of Problem \ref{mmseprob},  and   $\mathbf{F}(\infty)$, $\mathbf{K}(\infty)$ (corresponding to $\mathbf{U}^\ast$ in (\ref{mmserec}) according to (\ref{inv1})) is a stationary point of the problem in (\ref{utility}). Furthermore, Algorithm \ref{wmmsealg} converges to the unique global optimal point of  the problem in (\ref{utility}) for sufficiently small $L$.~\hfill\IEEEQED
\end{Lemma}
\begin{proof}
Please refer to Appendix F.		
\end{proof}

\subsection{Summary of the  Dynamic MIMO Precoder/Decorrelator Control and Performance Analysis}
\txtblue{Fig.~\ref{fig1wdassada} illustrates the signaling flow of Algorithm \ref{wmmsealg}. We  have the following remark discussing the signaling overhead  of the algorithm:}
\begin{Remark}	[\txtblue{Signaling Overhead  of Algorithm \ref{wmmsealg}}]	
\txtblue{Our proposed algorithm  has very low signaling overhead. To implement Algorithm 1, all the   BSs need to know the CSI matrices of the interference channels, and the QSI at the playback buffers of all the  mobile users. Specifically,}
	\begin{itemize}
		\item	{\txtblue{\textbf{CSI Signaling:}}} \txtblue{The knowledge of the CSI matrices at all the $K$  BSs can be achieved by each mobile user  $k$ broadcasting the local CSI measurements $\left\{\mathbf{H}_{kj}: \forall j \right\}$  to the $K$  BSs as illustrated in Fig.~\ref{fig1wdassada}. This CSI signaling requirement is the same as the conventional  interference mitigation schemes  in cooperative/coordinate MIMO  \cite{WMMSE}, \cite{zfbf}.}
		\item	{\txtblue{\textbf{QSI Signaling:}}} \txtblue{Besides the CSI signaling, an additional signaling requirement is the QSI. This can be achieved by each user  $k$ broadcasting  $Q_k$ to all the  $K$ BSs as illustrated in Fig.~\ref{fig1wdassadaasda}. Since $Q_k$  is a scalar, the additional signaling cost is negligible compared with the CSI signaling (which is a matrix feedback), and such scalar signaling can  be easily supported by the existing LTE measurement messages \cite{lte}. Furthermore, for constant playback rate at each mobile user, there is no need to explicitly feedback $Q_k$  to the BSs, because each BS $k$  can keep track of the transmit bits to the associated MS  $k$. Therefore, each BS  $k$ can maintain a virtual queue process as shown in Fig.~\ref{fig1wdassadab}, which has the same queue dynamics as the playback buffer at the mobile user.~\hfill\IEEEQED}				
	\end{itemize}
\end{Remark}

We  have the following remark discussing the complexity  of the Algorithm \ref{wmmsealg}:
\begin{Remark}	[Complexity Analysis of Algorithm \ref{wmmsealg}]
The computational complexity of Algorithm \ref{wmmsealg} is very low. Specifically, the complexity comes from computing the approximate value function in (\ref{finalapprox1}) and the precoding matrices in Algorithm \ref{wmmsealg}. The complexity of computing the closed-form approximate value function is very low compared with conventional  value iteration  methods \cite{DP_Bertsekas}. Computing the precoding matrices in Algorithm \ref{wmmsealg} is fast since each mobile user only needs to do twice matrix inversions (as in (\ref{inv1}) and (\ref{inv2})) and each BS needs to do one matrix inversions (as in (\ref{inv3})) based on local information at each time slot. Table.~\ref{tabletime} illustrates the comparison of the MATLAB computational time of the proposed solution, the baselines and the brute-force value iteration algorithm \cite{DP_Bertsekas}.~\hfill\IEEEQED
\end{Remark}

Finally, we analyze the  performance  gap between the optimal solution (by solving (\ref{OrgBel})) and the low complexity solution in Algorithm \ref{wmmsealg}.  Let $\widetilde{\Omega}$ represent the precoder and decorrelator  control  policy  in  Algorithm \ref{wmmsealg} and $\widetilde{\theta}=\limsup_{T \rightarrow \infty} \frac{1}{T} \sum_{t=0}^{T-1} \mathbb{E}^{\widetilde{\Omega}}  \left[c\left(\mathbf{Q} \left(t\right), \Omega\left(\boldsymbol{\chi}\left(t\right) \right)\right)    \right]$ be the associated average  performance. The performance gap between $\widetilde{\theta}$ and the optimal average  cost $\theta^{\ast}$ in (\ref{OrgBel}) is established in the following theorem:
\begin{Theorem}	[Performance Gap between $\widetilde{\theta}$ and $\theta^{\ast}$]	\label{perfgap}
	The performance gap between $\widetilde{\theta}$ and $\theta^{\ast}$ is given by
	\begin{align}		\label{perfgapequ}
		\widetilde{\theta} - \theta^{\ast} =\mathcal{O}(L)+o(1), \qquad \text{as } L \rightarrow 0, \tau  \rightarrow 0
	\end{align}~\hfill\IEEEQED
	\end{Theorem}
\begin{proof}
Please refer to Appendix G.		
\end{proof}

Theorem \ref{perfgap} suggests that  $\widetilde{\theta} \rightarrow \theta^{\ast} $, as $L \rightarrow 0$ and $\tau \rightarrow 0$. In other words, the proposed precoder and decorrelator  control  algorithm in Algorithm \ref{wmmsealg} is asymptotically optimal as $L \rightarrow 0$ and $\tau \rightarrow 0$.

\section{simulations}	\label{xxsim}
In this section, we compare the  performance  of the proposed  precoder and decorrelator  control scheme for multimedia streaming in Algorithm \ref{wmmsealg}  with the following three baselines using numerical simulations:
\begin{itemize}
\item \textbf{Baseline 1, Zero-Forcing Precoding (ZFP) \cite{zfbf}}: The $K$ BSs adopt zero-forcing precoding matrix and fixed power transmission at each time slot. The precoding matrix of BS $k$ is obtained by projection of $\mathbf{H}_{kk}$ on the orthogonal complement of the subspace span $\left( \left[\mathbf{H}_{jk}\right]_{j \neq k}\right)$.

\item \textbf{Baseline 2, CSI-Only  Precoding (COP) \cite{WMMSE}}: The precoding matrix of each BS is obtained by solving the following  problem at each time slot: $\min _{\mathbf{F},\mathbf{U} }  \sum_{k=1}^K  \left(\mathrm{Tr}\left(  \mathbf{F}_k  \mathbf{F}_k^\dagger\right)  -   \alpha  R_k\left(\mathbf{H}, \mathbf{F}, \mathbf{U}_k\right)  \right)$ for all $k$,  where  $\alpha$ is used to adjust the tradeoff between the transmit power and the data rate. The optimal CSI-only precoding control   is only adaptive to CSI.

\item \textbf{Baseline 3, Queue-Weighted Precoding (QWP) \cite{lya1}}: The precoding matrix of each BS is obtained by solving the following  problem at each time slot: $\min _{\mathbf{F},\mathbf{U} }  \sum_{k=1}^K  \left(\mathrm{Tr}\left(  \mathbf{F}_k  \mathbf{F}_k^\dagger\right)  -   \alpha  [Q^h-Q_k]^+  \right. \\ \left.  \times R_k\left(\mathbf{H}, \mathbf{F}, \mathbf{U}_k\right)  \right)$ for all $k$. The optimal queue-weighted precoding control   is  adaptive to CSI and QSI.
\end{itemize}

\begin{figure}
\begin{minipage}[t]{0.45\textwidth}
  \centering
  \includegraphics[width=3.5in]{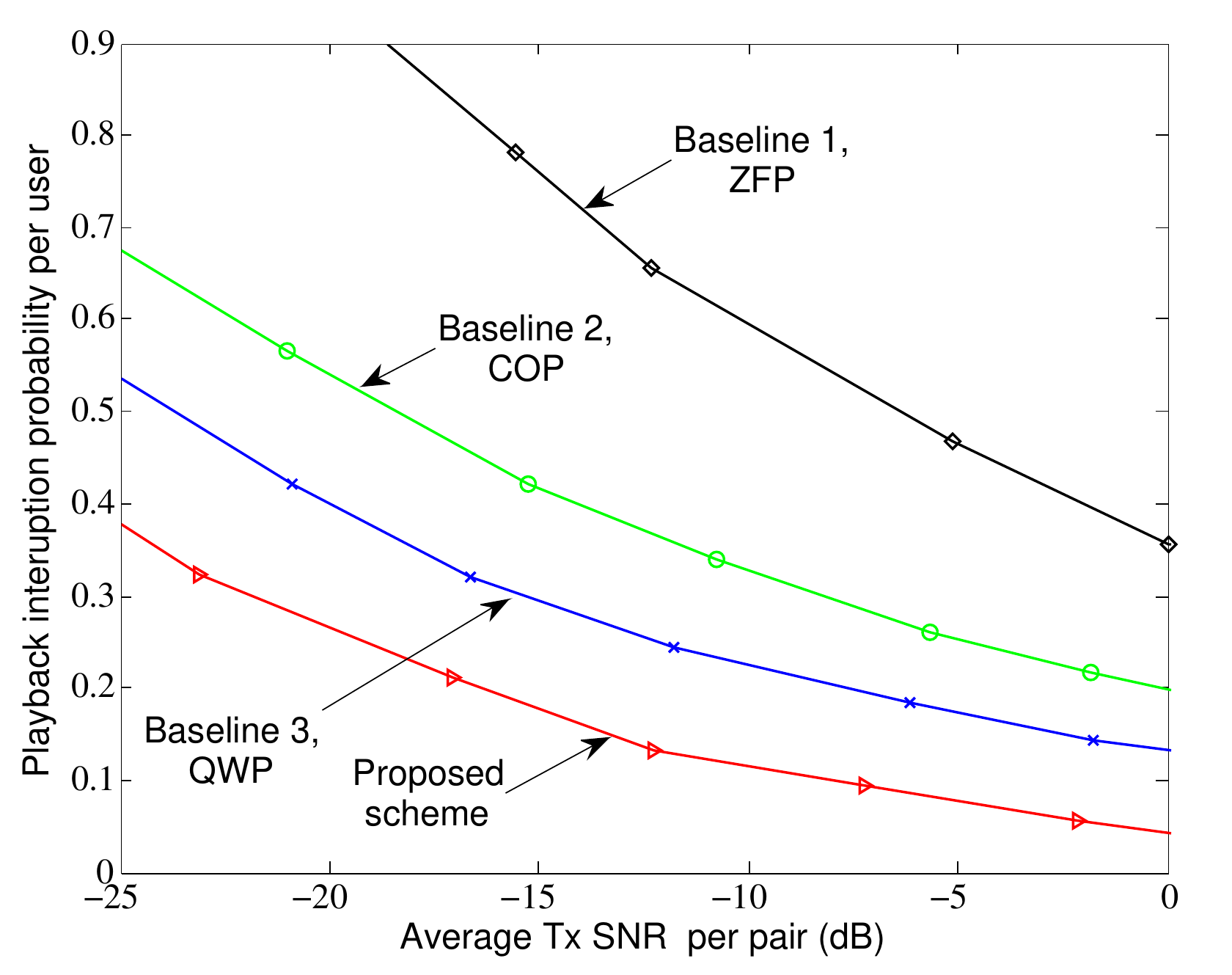}
  \caption{\txtblue{Playback interruption probability per user versus average Tx SNR per pair, with $K=5$, $N_t=5$ and $N_r=2$.}}
  \label{fig1}
\end{minipage}
\hspace{0.05\textwidth}
\begin{minipage}[t]{0.45\textwidth}
  \centering
   \includegraphics[width=3.5in]{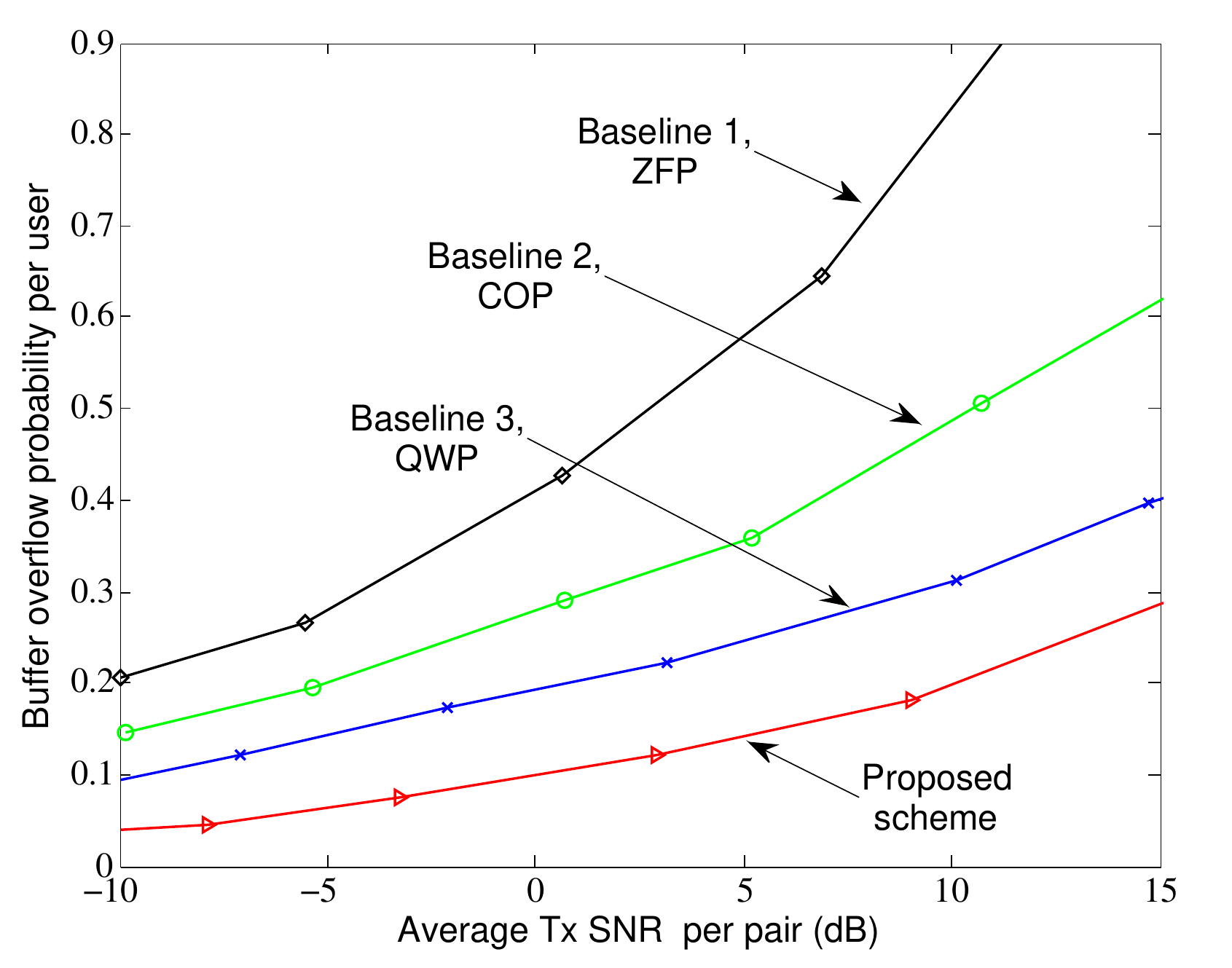}
  \caption{\txtblue{Buffer overflow probability per user versus average Tx SNR per pair, with  $K=5$, $N_t=5$ and $N_r=2$.}}
  \label{fig2}\end{minipage}	
 \end{figure}

In the simulations, we consider multimedia streaming in a $K$-pair MIMO interference network \txtblue{under the 802.11e WLAN setup as in \cite{apvideo}}. The channel fading coefficient and the channel noise are complex Gaussian distributed. For the direct and the cross channel long-term path gain, we let $\frac{L_{kj}}{L_{kk}}=0.1$ for all $k \neq j$ as in \cite{simtopo}.  We consider constant bit rate video streaming for each mobile user with  streaming rate  equal to \txtblue{$1.5$ Mbps as in \cite{apvideo}}. The decision slot duration $\tau$ is $10$ ms. The total bandwidth is $1$MHz. Furthermore, we let $\gamma_k=\beta_k=\beta$ for all $k$ and vary $\beta$ to obtain different tradeoff curves. The mobile users adopt MMSE decorrelator as in (\ref{mmserec}) for all the baselines. We consider the  average power cost (\ref{delay_cost}), playback interruption probability (\ref{ind1}) and buffer overflow probability (\ref{ind2}) as the performance metrics for each multimedia streaming flow. The other system parameters are configured as: $\eta=50$, \txtblue{$Q^l=50$ Kbits and $Q^h=150$ Kbits}.

\begin{figure}
\begin{minipage}[t]{0.45\textwidth}
\centering
  \includegraphics[width=3.5in]{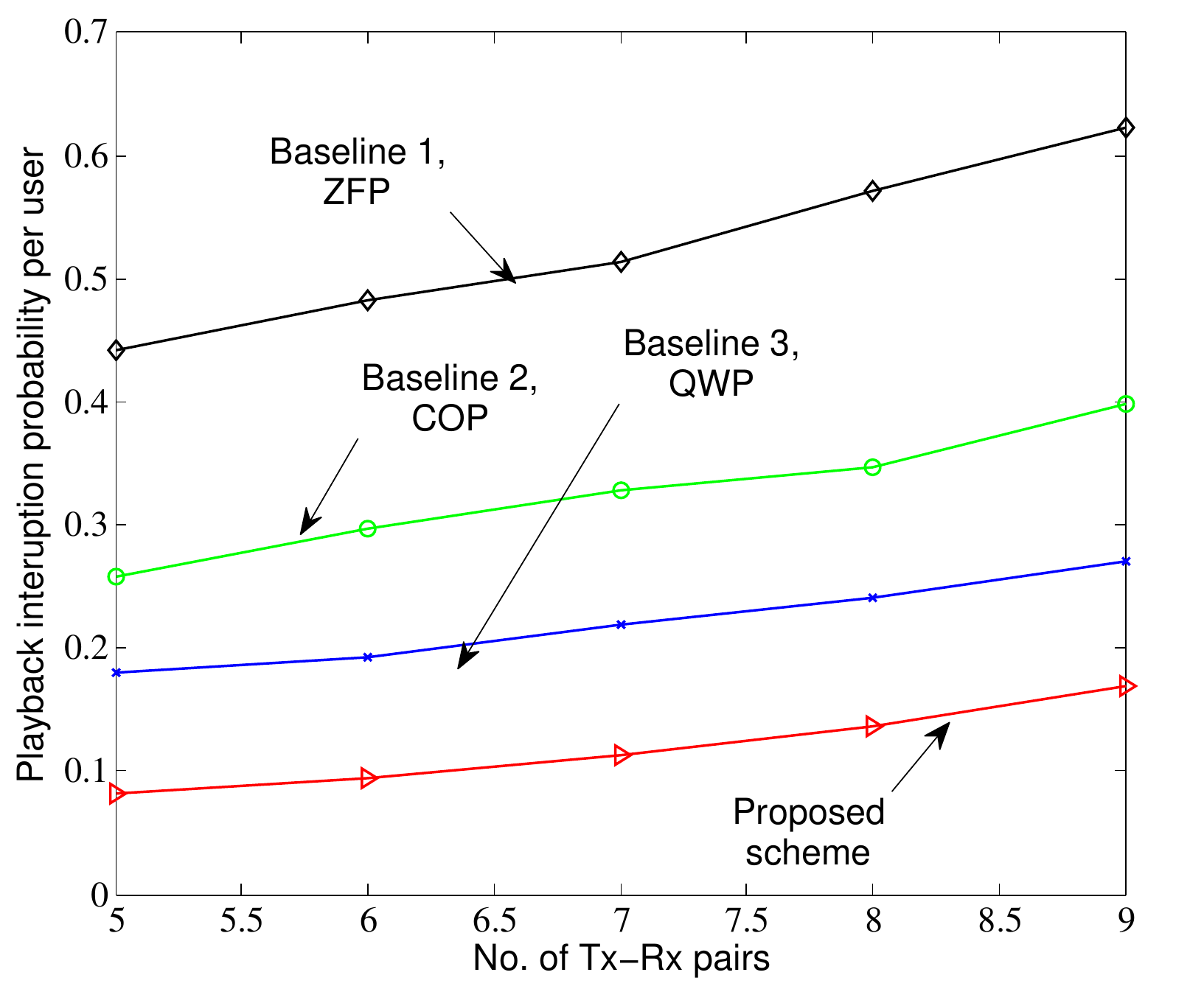}
  \caption{\txtblue{Playback interruption probability per user versus no. of  Tx-Rx pairs at average transmit SNR $ =-5$ dB.}}
  \label{fig3}
\end{minipage}
\vspace{0.5cm}
\hspace{0.05\textwidth}
\begin{minipage}[t]{0.45\textwidth}
  \centering
   \includegraphics[width=3.5in]{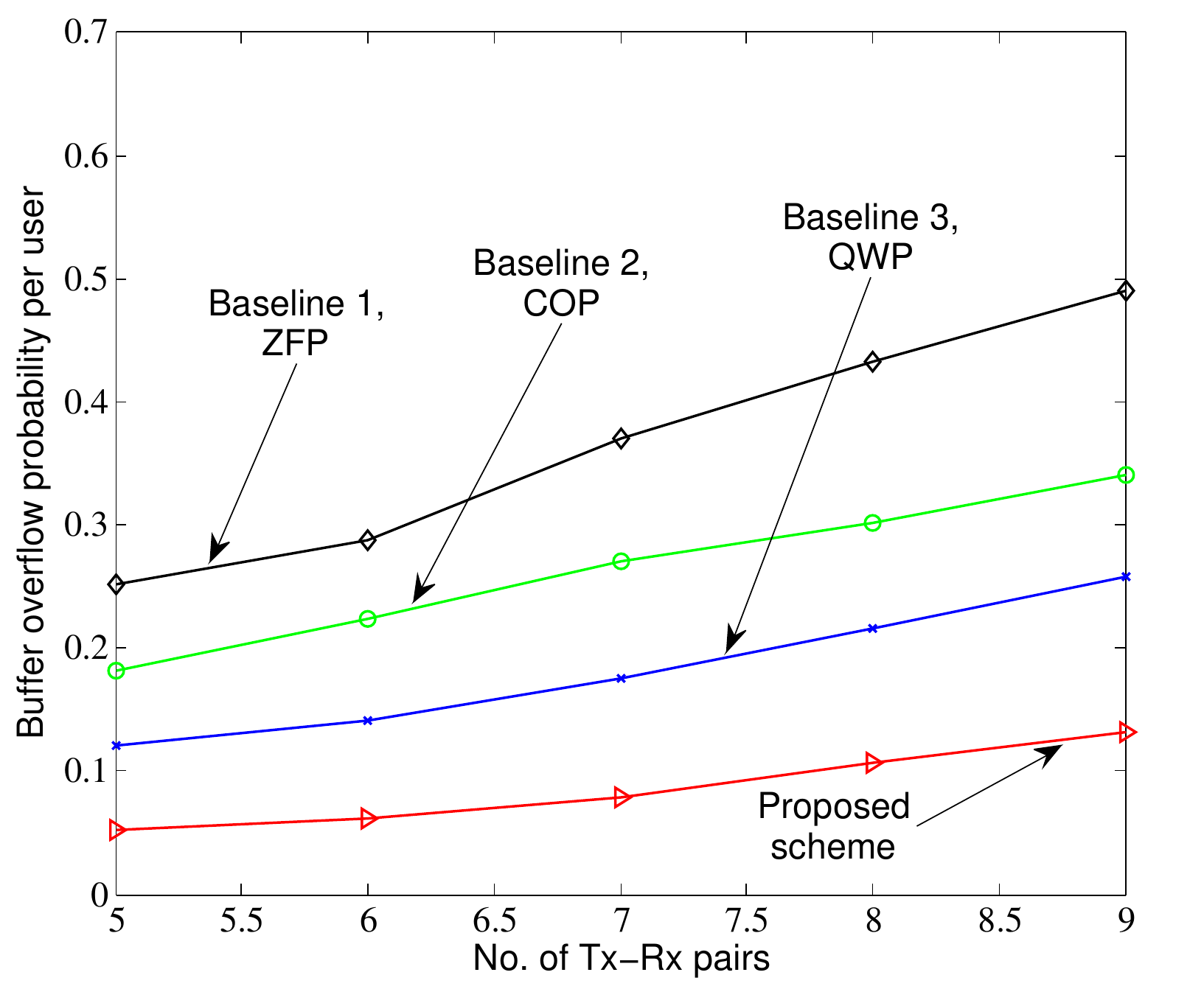}
  \caption{\txtblue{Buffer overflow probability per user versus no. of  Tx-Rx pairs at average transmit SNR $ =-5$ dB.}}
  \label{fig4}\end{minipage}	
\end{figure}

\subsection{Playback Interruption and Buffer Overflow Probabilities versus Average Transmit SNR}

Fig.~\ref{fig1} and Fig.~\ref{fig2} illustrates the playback interruption probability  and buffer overflow probability per user versus average transmit SNR per pair.  The proposed scheme achieves significant performance gain over all the baselines across a wide range of SNR values. It can also be observed that there exists a tradeoff between  the  playback interruption   and buffer overflow probabilities, and we cannot decrease them both by adjusting the transmit power. From Fig.~\ref{fig1} and Fig.~\ref{fig2}, we can see that the best SNR region for using our proposed  algorithm is around -5 dB, where both the playback interruption and buffer overflow probabilities are relatively low, \txtblue{and there is also significant performance gain over the baselines.}

\subsection{Playback Interruption and Buffer Overflow Probabilities versus Number of Tx-Rx Pairs}

Fig.~\ref{fig3} and Fig.~\ref{fig4} illustrates the playback interruption probability  and buffer overflow probability per user  versus the number of Tx-Rx pairs.  The number of transmit antennas at the BS is $K$ (which is equal to the number of Tx-Rx pairs) and the number of receive antennas at the mobile users is $2$. It can be observed that our proposed scheme achieves  significant performance gain over  all the baselines across a wide range of the numbers of Tx-Rx pairs. \txtblue{\subsection{Performance Comparison  of the Proposed Algorithm and the Optimal Solution}
Fig.~\ref{fig30}  illustrates the playback interruption and buffer overflow probabilities per user  versus the carrier sensing distance $\delta$  for both the proposed Algorithm \ref{wmmsealg} and the brute-force value iteration (VIA) algorithm\footnote{Note that the brute-force VIA \cite{DP_Bertsekas} solves the discrete time Bellman equation in (\ref{OrgBel}) and gives the optimal average cost.} \cite{DP_Bertsekas}. It can be observed that the performance of our proposed algorithm is very close to that of the brute-force VIA algorithm and the performance gap becomes smaller as  $\delta$ increases\footnote{\txtblue{As $\delta$ increases, the worst-case cross channel path gain decreases according to the path loss model in Section \ref{frisss}.}}. This is in accordance with the performance gap analysis in Theorem 4.}

\begin{figure}
\centering
  \includegraphics[width=3.5in]{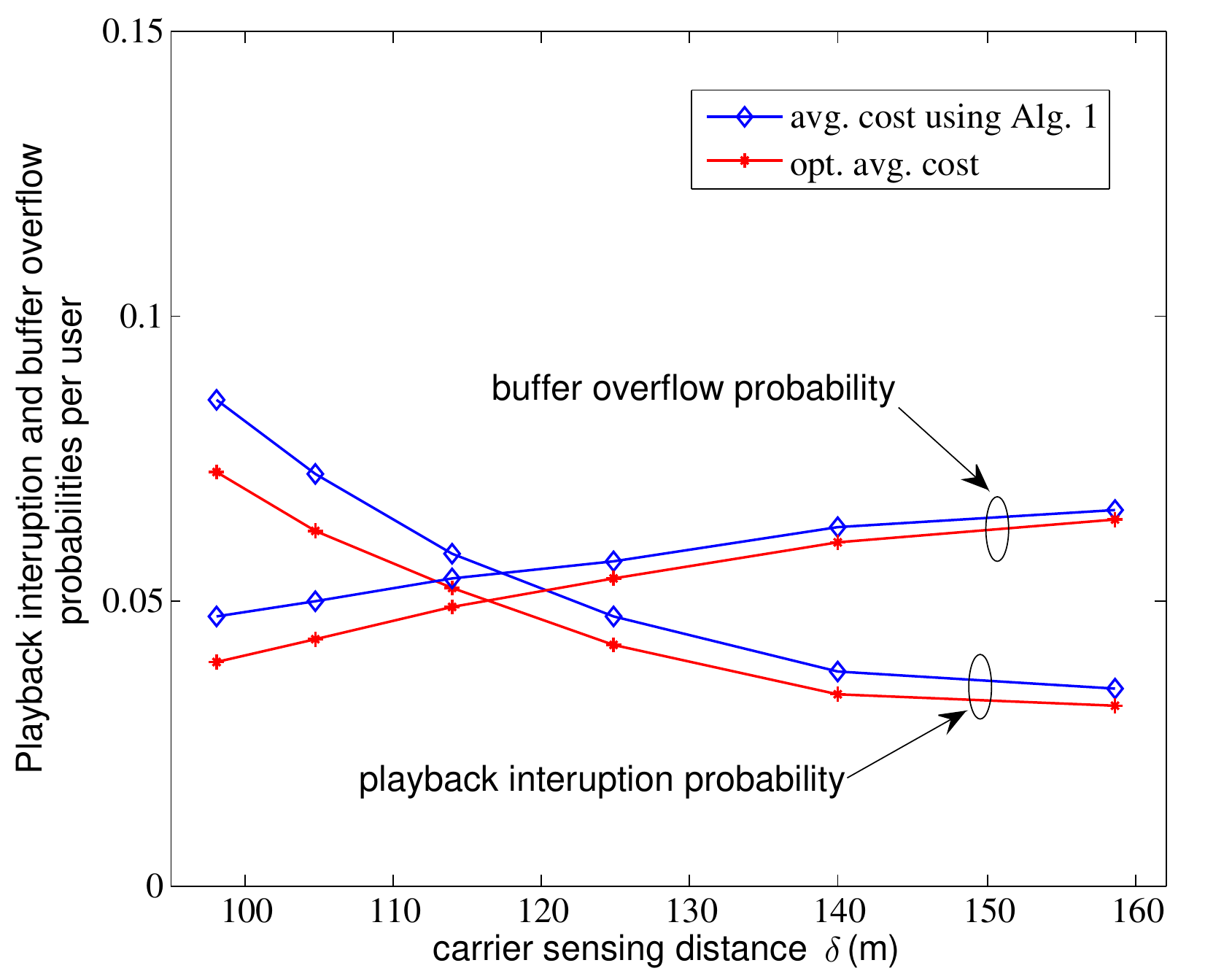}
  \caption{\txtblue{Playback interruption and buffer overflow probabilities per user versus the  carrier sensing distance at average transmit SNR $ =-5$ dB, with  $K=4$, $N_t=4$ and $N_r=2$. The path loss model is given  in Section \ref{frisss}, where  $G_r=G_t=3 $ dB and $\lambda=0.125$ m (2.4 GHz carrier frequency). The direct channel path gain is $-75$ dBm.}}
  \label{fig30}
\end{figure}

\subsection{Computational Complexity Analysis}
Table \ref{tabletime}  illustrates the comparison of the MATLAB computational time of the proposed solution, the  baselines and the brute-force value iteration algorithm \cite{DP_Bertsekas}. \txtblue{Note that the computation time of all the algorithms  increases as $K$ increases, and this is a fair price to pay.  The computational time of Baseline 1  is the smallest in all different $K$ scenarios, but it has the worst performance. On the other hand, our proposed solution has a similar order of  complexity growth w.r.t. $K$ compared with Baseline 2/3 , but the proposed solution has much better performance.}

\begin{table}
	\centering
\begin{tabular}{|c|c|c|c|c|c|}
	\hline
		  &${K=4}$ & $ {K=6}$  & $ {K=8}$ & $ {K=10}$   \\
	\hline
		 {Baseline 1, ZFP} &  {0.002s}&  {0.003s} &  {0.004s} &  {0.006s}    			\\
		 {Baseline 2/3, COP/QWP}&{0.014s} &   {0.023s} &  {0.035s} &    {0.051s}	 \\
	            {Proposed Scheme} &   {0.052s} &  {0.080s} &  {0.125s}  &   {0.191s}	\\	
	            {Value Iteration Algorithm}   &    {657s}  &    {$>10^4$s} &  {$>10^4$s}  &   {$>10^4$s} 	\\
	\hline
\end{tabular}
	\caption{\txtblue{Comparison of the MATLAB computational time of the proposed scheme, the baselines and  the value iteration algorithm in one decision slot.}}	
		\label{tabletime}
\end{table}

\section{Summary}
In this paper, we propose an asymptotically optimal dynamic precoder/decorrelator control  to support multimedia streaming applications in MIMO interference networks. We formulate the associated stochastic optimization problem as an infinite horizon average cost MDP  and derive the  sufficient conditions for optimality. Using the weak interference property of the wireless network, we derive a closed-form approximate value function to the $K$-dimensional optimality equation and the associated error bound using perturbation analysis. Based on the closed-form approximate value function, we propose an asymptotically optimal low complexity precoder/decorrelator  control algorithm and establish the performance  gap between the optimal solution and the proposed low complexity solution. Numerical results show that the proposed  scheme has much better  performance than the other  baselines.

\section*{Appendix A: Proof of Theorem \ref{LemBel}}
Following \emph{Proposition 4.6.1} of \cite{DP_Bertsekas},  the sufficient conditions for optimality of Problem \ref{IHAC_MDP} is that there exists a ($\theta^\ast, \{ V^\ast\left(\mathbf{Q}  \right) \}$) that satisfies the following Bellman equation and $V^\ast$ satisfies the transversality condition in (\ref{transodts})  for all  admissible control policy $\Omega$ and initial  state $\mathbf{Q} \left(0 \right)$:
\begin{small}\begin{align}
	&\theta^\ast\tau + V^\ast\left(\boldsymbol{\chi} \right) 	 \\
	=& \min_{\mathbf{F}, \mathbf{U}} \Big[ c\left( \mathbf{Q}, \mathbf{F}\right)\tau+  \sum_{\boldsymbol{\chi}' } \Pr\big[ \boldsymbol{\chi}'\big| \boldsymbol{\chi}, \mathbf{F}, \mathbf{U}\big]  V^\ast\left(\boldsymbol{\chi}'\right)    \Big]	\notag \\
	=& \min_{\mathbf{F}, \mathbf{U}} \Big[ c\left(\mathbf{Q}, \mathbf{F}\right)\tau+ \sum_{\mathbf{Q} '} \sum_{\mathbf{H}' }  \Pr \big[ \mathbf{Q} '\big| \boldsymbol{\chi}, \mathbf{F}, \mathbf{U} \big] \Pr \big[\mathbf{H}' \big]  V^\ast\left(\boldsymbol{\chi}'\right)    \Big]\notag	
\end{align}\end{small}Taking expectation w.r.t. $\mathbf{H}$ on both sizes of the above equation and denoting $V^\ast\left(\mathbf{Q}  \right) = \mathbb{E}\big[V^\ast\left(\boldsymbol{\chi} \right) \big| \mathbf{Q} \big]$, we obtain the equivalent Bellman equation in (\ref{OrgBel}) in Theorem \ref{LemBel}.

\section*{Appendix B: Proof of Corollary \ref{cor1}}
Let $\mathbf{Q} ' =(Q_1',\cdots, Q_k')= \mathbf{Q} (t+1)$ and $\mathbf{Q} =(Q_1,\cdots, Q_k)=\mathbf{Q} (t)$. For the queue dynamics in (\ref{Qdyn}) and sufficiently small $\tau$, we have $Q_k'  =  Q_k- \mu_k \tau + R_k\left(\mathbf{H}, \mathbf{F},\mathbf{U}_k  \right) \tau,  \forall k$. Therefore, if  $V\left(\mathbf{Q} \right)$ is of class $\mathcal{C}^2(\mathbb{R}_+^K)$, we have the following Taylor expansion on $V\left( \mathbf{Q} '\right)$ in (\ref{OrgBel}): \bs$\mathbb{E}\left[ V\left( \mathbf{Q} '\right) \big| \mathbf{Q}  \right] =V\left( \mathbf{Q} \right)+\sum_{k=1}^K  \frac{\partial V\left(\mathbf{Q} \right)}{\partial Q_k} \left[  \mathbb{E}\left[R_k\left(\mathbf{H}, \mathbf{F},\mathbf{U}_k  \right)\big| \mathbf{Q} \right] -\mu_k  \right]\tau + o(\tau)$.\bsc For notation convenience, let $T_{\boldsymbol{\chi}}(\theta, V, \mathbf{F}, \mathbf{U})$ and $T_{\boldsymbol{\chi}}^\dagger(\theta, V, \mathbf{F}, \mathbf{U})$ denote the \emph{Bellman operators}:
\bs \begin{align}	\label{beloperator1}
	  T_{\boldsymbol{\chi}}(\theta, V, \mathbf{F}, \mathbf{U})=T_{\boldsymbol{\chi}}^\dagger(\theta, V, \mathbf{F}, \mathbf{U})+\nu  G_{\boldsymbol{\chi}}(V,\mathbf{F}, \mathbf{U})
\end{align}\bsc
for some smooth function $G_{\boldsymbol{\chi}}$ and $\nu=o(1)$ (which asymptotically goes to zero as $\tau$ goes to zero), and denote
\bs\begin{align}
	 &T_{\boldsymbol{\chi}}^\dagger(\theta, V, \mathbf{F}, \mathbf{U}) = -\theta + {c}\left(\mathbf{Q},  \mathbf{F}\right) \notag \\
	 &\hspace{2.5cm}+  \sum_{k=1}^K \frac{\partial V \left(\mathbf{Q}  \right) }{\partial Q_k} \left[ R_k\left(\mathbf{H},  \mathbf{F}, \mathbf{U}\right)  -\mu_k  \right]		\\
	 &T_{\boldsymbol{\chi}}(\theta, V)=\min_{  \mathbf{F}, \mathbf{U}} T_{\boldsymbol{\chi}}(\theta, V, \mathbf{F}, \mathbf{U}),  \notag \\
	 &T_{\boldsymbol{\chi}}^\dagger(\theta, V)=\min_{  \mathbf{F}, \mathbf{U}} T_{\boldsymbol{\chi}}^\dagger(\theta, V, \mathbf{F}, \mathbf{U})	\label{zerofuncsa}
\end{align}\bsc
Suppose $(\theta^\ast,V^\ast)$ satisfies the Bellman equation in (\ref{OrgBel}) and $(\theta,V)$ satisfies the approximate Bellman equation in (\ref{bellman2}), we have for any $\mathbf{Q}  \in \boldsymbol{\mathcal{Q}}$,
\bs\begin{align}
	\mathbb{E}\left[T_{\boldsymbol{\chi}}(\theta^\ast, V^\ast)|\mathbf{Q}\right]=0, \quad \mathbb{E}\left[T_{\boldsymbol{\chi}}^\dagger(\theta, V)|\mathbf{Q}\right]=0	\label{zerofunc}
\end{align}\bsc

Then, we establish the following lemma.
\begin{Lemma}	\label{applemma}
$\big|\mathbb{E}\left[T_{\boldsymbol{\chi}}  (\theta, V)\big|\mathbf{Q}\right]\big|=o(1)$,  $\forall {\boldsymbol{\chi}} $, where $o(1)$ asymptotically goes to zero as $\tau$ goes to zero.~\hfill\IEEEQED
\end{Lemma}

\begin{proof}	[Proof of Lemma \ref{applemma}]
For any $\boldsymbol{\chi}$, we have \bs$T_{\boldsymbol{\chi}} (\theta, V)=\min_{ \mathbf{F}, \mathbf{U}}\left[ T_{\boldsymbol{\chi}}^\dagger(\theta, V, \mathbf{F}, \mathbf{U})+\nu  G_{\boldsymbol{\chi}}(V,\mathbf{F}, \mathbf{U}) \right] \geq \min_{ \mathbf{F}, \mathbf{U}} T_{\boldsymbol{\chi}}^\dagger(\theta, V, \mathbf{F}, \mathbf{U}) + \nu \min_{ \mathbf{F}, \mathbf{U}} G_{\boldsymbol{\chi}}(V,\mathbf{F}, \mathbf{U})$\bsc. On the other hand, \bs$T_{\boldsymbol{\chi}} (\theta, V) \leq \min_{ \mathbf{F}, \mathbf{U}} T_{\boldsymbol{\chi}}^\dagger(\theta, V, \mathbf{F}, \mathbf{U}) + \nu G_{\boldsymbol{\chi}}(V,\mathbf{F}^\dagger, \mathbf{U}^\dagger)$\bsc, where \bs$(\mathbf{F}^\dagger, \mathbf{U}^\dagger)= \arg \min_{\mathbf{F}, \mathbf{U}} T_{\boldsymbol{\chi}}^\dagger(\theta, V, \mathbf{F}, \mathbf{U}) $\bsc.

From (\ref{zerofuncsa}) and (\ref{zerofunc}), \bs$\mathbb{E}\left[\min_{ \mathbf{F}, \mathbf{U}} T_{\boldsymbol{\chi}}^\dagger(\theta, V, \mathbf{F}, \mathbf{U})\big|\mathbf{Q}\right]=\mathbb{E}\left[T_{\boldsymbol{\chi}}^\dagger(\theta, V)\big|\mathbf{Q}\right]=0$\bsc. Since \bs$T_{\boldsymbol{\chi}}^\dagger(\theta, V, \mathbf{F}, \mathbf{U}) $\bsc and \bs$G_{\boldsymbol{\chi}}(V,\Omega^\dagger(\mathbf{Q} )) $\bsc  are all smooth and bounded functions, we have $\big|\mathbb{E}\left[T_{\boldsymbol{\chi}}  (\theta, V)\big|\mathbf{Q}\right]\big| = \mathcal{O}(\nu)=o(1)$ for any $\mathbf{Q}  \in \boldsymbol{\mathcal{Q}}$, where $o(1)$ asymptotically goes to zero as $\tau$ goes to zero.
\end{proof}

Finally, we  prove the final result as follows.
\begin{Lemma}		\label{tenlemma}
	Suppose \bs$\mathbb{E}[T_{\boldsymbol{\chi}}(\theta^\ast, V^\ast)|\mathbf{Q}] = 0$\bsc for all $\mathbf{Q} $ together with the transversality condition in (\ref{transodts})  has a unique solution $(\theta^*, V^\ast)$. If $(\theta, V)$ satisfies the approximate Bellman equation in (\ref{bellman2}) and the transversality condition in (\ref{transodts}), then $|\theta-\theta^\ast|=o\left(1 \right)$, $|V\left(\mathbf{Q}  \right)-V^\ast\left(\mathbf{Q}  \right)|=o\left(1 \right)$ for all  $\mathbf{Q}  $, where the error term $o(1)$ asymptotically goes to zero  as $\tau$ goes to zero.~\hfill\IEEEQED	
\end{Lemma}
\begin{proof}	[Proof of Lemma \ref{tenlemma}]
	Suppose for some $\mathbf{Q} '$, we have $V\left(\mathbf{Q} ' \right)=V^\ast\left(\mathbf{Q} ' \right)+\alpha$ for some $\alpha \neq 0$. From Lemma \ref{applemma}, we have $\big|\mathbb{E}[T_{\boldsymbol{\chi}}(\theta, V)|\mathbf{Q}]\big|= o(1)$ for all $\mathbf{Q} $. Now let $\tau \rightarrow 0$, we have $(\theta, V)$ satisfies $\mathbb{E}[T_{\boldsymbol{\chi}}(\theta, V)|\mathbf{Q}]= 0$ for all $\mathbf{Q} $ and the  transversality condition in (\ref{transodts}). However, $V\left(\mathbf{Q} ' \right) \neq V^\ast\left(\mathbf{Q} ' \right)$ because of the assumption that $V\left(\mathbf{Q} ' \right)=V^\ast\left(\mathbf{Q} ' \right)+\alpha$. This contradicts  the condition that $(\theta^*, V^\ast)$ is a unique solution of $\mathbb{E}\left[T_{\boldsymbol{\chi}}(\theta^\ast, V^\ast)|\mathbf{Q}\right]=0$ for all $\mathbf{Q} $  and the transversality condition in (\ref{transodts}). Hence, we must have $|V\left(\mathbf{Q}  \right)-V^\ast\left(\mathbf{Q}  \right)|=o\left(1 \right)$ for all  $\mathbf{Q}  $, where $o(1)$ asymptotically goes to zero as $\tau$ goes to zero. Similarly, we can establish $|\theta- \theta^\ast| = o(1)$.
\end{proof}

\section*{{Appendix C: Proof of Theorem \ref{HJB1}}}	
For simplicity of notation, we write $J\left(\mathbf{Q} \right)$ in place of $J\left(\mathbf{Q} ;L \right)$. We first establish the relationship between $J\left(\mathbf{Q} \right)$ and  $V \left(\mathbf{Q}  \right)$. We can observe that if ($c^{\infty}, \{ J\left(\mathbf{Q}  \right) \}$) satisfies  the PDE in (\ref{cenHJB}), it also satisfies the approximate Bellman equation in (\ref{bellman2}). Furthermore, since $J\left(\mathbf{Q} \right)=\mathcal{O}(\sum_{k=1}^K Q_k)$, we have $\lim_{t \rightarrow \infty}\mathbb{E}^{\Omega}\left[J\left(\mathbf{Q} (t)\right) \right]< \infty$ for any admissible policy $\Omega$. Hence, $J\left(\mathbf{Q}  \right)=\mathcal{O}(\sum_{k=1}^K Q_k)$ satisfies the transversality condition in (\ref{transodts}).

Next, we show that the optimal control policy $\Omega^{J \ast}$ obtained by solving the PDE in (\ref{cenHJB}) is an admissible control policy in the discrete time system as defined in Definition \ref{adddtdomain}.

We first establish the following lemma:
\begin{Lemma}	\label{lemmaappr}
	For any $k$, if $\frac{\partial  J \left(\mathbf{Q} ;{L}\right)} {\partial Q_k}>0$ for sufficient large $Q_k$, then the optimal precoding matrix for user $k$obtained by solving (\ref{cenHJB}) is  $\mathbf{F}_k^{J\ast}=\mathbf{0}$.~\hfill\IEEEQED
\end{Lemma}
\begin{proof}	[Proof of Lemma \ref{lemmaappr}]
	For for sufficient large $Q_k$, let $\mathcal{I}_k=\big\{i: \frac{\partial  J \left(\mathbf{Q} ;{L}\right)} {\partial Q_i}>0\big\}$ ($k \in \mathcal{I}_k$), and let $\mathcal{I}_k^c=\big\{i: \frac{\partial  J \left(\mathbf{Q} ;{L}\right)} {\partial Q_i} \leq 0\big\}$. Let $\mathbf{F}^{J \ast}=\big\{\mathbf{F}_i^{J\ast}: \forall i \big\}$ be the optimal precoding matrices obtained by solving the PDE in (\ref{cenHJB}). Suppose  some $\mathbf{F}_i^{J\ast} \neq \mathbf{0}$ ($i \in \mathcal{I}_k' \subset \mathcal{I}_k$). Denote $\widetilde{\mathbf{F}}^{J \ast}=\big\{\widetilde{\mathbf{F}}_i^{J\ast}=\mathbf{F}_i^{J\ast}:  i \notin  \mathcal{I}_k'\big\} \cup \big\{\widetilde{\mathbf{F}}_i^{J\ast} = \mathbf{0}:  i \in  \mathcal{I}_k'  \big\}$. Denote the objective function  in (\ref{cenHJB}) for given $\boldsymbol{\chi}$ as 
	\begin{small}\begin{align}
	f_{\boldsymbol{\chi}}\left( \mathbf{F}\right)=\sum_{i=1}^K \left(\mathrm{Tr}\left(  \mathbf{F}_i  \mathbf{F}_i^\dagger\right)+\frac{\partial  J \left(\mathbf{Q} ;{L}\right)} {\partial Q_i} \left(R_i\left(\mathbf{H}, \mathbf{F}\right)  \right) \right)
	\end{align} \end{small}where the optimal MMSE receiver \cite{WMMSE} is adopted at the receiver and \bs$R_i(\mathbf{H}, \mathbf{F}) =W \log_2 \det \big( \mathbf{I}  + L_{ii} \mathbf{H}_{ii}\mathbf{F}_i \mathbf{F}_i^\dagger \mathbf{H}_{ii}^\dagger \big(\sum_{j \neq i} L_{ij}\mathbf{H}_{ij}\mathbf{F}_j \mathbf{F}_j^\dagger \mathbf{H}_{ij}^\dagger+\mathbf{I} \big)^{-1}\big)$\bsc. Then, we have
	\begin{small}\begin{align}
		& f_{\boldsymbol{\chi}}( \mathbf{F}^{J \ast})	\label{44equ}	\\
		=& \sum_{i\in \mathcal{I}_k} \Big(\mathrm{Tr}\left(  \mathbf{F}_i^{J \ast} ( \mathbf{F}_i^{J \ast})^\dagger\right)+\frac{\partial  J \left(\mathbf{Q} ;{L}\right)} {\partial Q_i} \left(R_i\left(\mathbf{H}, \mathbf{F}^{J \ast}\right)  \right) \Big)\notag \\
		&+\sum_{i\notin \mathcal{I}_k} \Big(\mathrm{Tr}\left(  \mathbf{F}_i^{J \ast} ( \mathbf{F}_i^{J \ast})^\dagger\right)+\frac{\partial  J \left(\mathbf{Q} ;{L}\right)} {\partial Q_i} \left(R_i\left(\mathbf{H}, \mathbf{F}^{J \ast} \right)  \right) \Big)	\notag \\
		\overset{(a)}{\geq} & \sum_{i\notin \mathcal{I}_k} \Big(\mathrm{Tr}\left(  \mathbf{F}_i^{J \ast} ( \mathbf{F}_i^{J \ast})^\dagger\right)+\frac{\partial  J \left(\mathbf{Q} ;{L}\right)} {\partial Q_i} \left(R_i\left(\mathbf{H}, \mathbf{F}^{J \ast} \right)  \right) \Big)	\notag \\
		\overset{(b)}{>} & \sum_{i\notin \mathcal{I}_k} \left(\mathrm{Tr}\left( \widetilde{ \mathbf{F}}_i^{J \ast} ( \widetilde{\mathbf{F}}_i^{J \ast})^\dagger\right)+\frac{\partial  J \left(\mathbf{Q} ;{L}\right)} {\partial Q_i} \left(R_i\left(\mathbf{H}, \widetilde{\mathbf{F}}^{J \ast} \right)  \right) \right)	\notag \\
		=& f_{\boldsymbol{\chi}}( \widetilde{\mathbf{F}}^{J \ast})	\notag
	\end{align}\end{small}where $(a)$ is due to $ \mathrm{Tr}(  \mathbf{F}_i^{J \ast} ( \mathbf{F}_i^{J \ast})^\dagger)+\frac{\partial  J (\mathbf{Q} ;{L})} {\partial Q_i} (R_i(\mathbf{H}, \mathbf{F}^{J \ast})  )  \geq 0$ for $i\in \mathcal{I}_k$, $(b)$ is due to $\mathrm{Tr}(  \mathbf{F}_i^{J \ast} ( \mathbf{F}_i^{J \ast})^\dagger)=\mathrm{Tr}( \widetilde{ \mathbf{F}}_i^{J \ast} ( \widetilde{\mathbf{F}}_i^{J \ast})^\dagger)$  and $R_i(\mathbf{H}, \mathbf{F}^{J \ast} ) \leq R_i(\mathbf{H}, \widetilde{\mathbf{F}}^{J \ast} )$ for $i \notin \mathcal{I}_k$. Therefore, from (\ref{44equ}), $\widetilde{\mathbf{F}}^{J \ast}$ achieves smaller objective than $ \mathbf{F}^{J \ast}$, which contradicts that $ \mathbf{F}^{J \ast}$ is the optimal solution.  Therefore, for $i \in \mathcal{I}_k$ ($k \in \mathcal{I}_k$), the optimal precoding matrix is $\mathbf{0}$. 
\end{proof}

Define the \emph{semi-invariant moment generating function} of $R_k\big(\mathbf{H},\Omega^{J \ast}(\boldsymbol{\chi})\big)-\mu_k$ as $\phi_k(r,\mathbf{Q} )=  \ln \big(\mathbb{E}\big[e^{\left(R_k(\mathbf{H},\Omega^{J \ast}(\boldsymbol{\chi})  )-\mu_k \right)r}\big| \mathbf{Q}  \big] \big)$. According to Lemma \ref{lemmaappr}, we have $\mathbb{E}\big[R_k(\mathbf{H},\Omega^{J \ast}(\boldsymbol{\chi})) -\mu_k  \big|\mathbf{Q} \big]=-\mu_k  < 0$ when $Q_k > \overline{Q}_k$ for some large $\overline{Q}_k$. Hence,  $\phi_k(r,\mathbf{Q} )$ will have a unique positive root $r_k^\ast(\mathbf{Q} )$ ($\phi_k(r_k^\ast(\mathbf{Q} ),\mathbf{Q} )=0$) \cite{lya2}.  Let $r_k^\ast= r_k^\ast(\overline{\mathbf{Q} })$, where $\overline{\mathbf{Q} }=(\overline{Q}_1, \dots, \overline{Q}_K)$. We then have the following lemma on the tail distribution $Q_k$, $\Pr\big[ Q_k\geq x\big]$.
\begin{Lemma}	[Kingman Bound \cite{lya2}]	\label{kingmanres}
	$P_k(x) \triangleq \Pr\big[ Q_k \geq x \big] \leq e^{-r_k^\ast x} $, if $x \geq \overline{x}_k$ for sufficiently large $\overline{x}_k$.~\hfill\IEEEQED
\end{Lemma}

Finally, we check whether $\Omega^{J \ast}$ stabilizes the system according to the definition of the admissible control policy in Definition \ref{adddtdomain} as follows: \bs$\mathbb{E}^{\Omega^{J \ast}} \left[J\left(\mathbf{Q} \right) \right] \leq C \sum_{k=1}^K \mathbb{E}^{\Omega^{J \ast}} \left[ Q_k  \right]= C\sum_{k=1}^K \left[\int_0^{\infty} \Pr \left[Q_k >s \right] \mathrm{d}s \right] \leq  C \sum_{k=1}^K \left[\overline{x}_k+ \int_{\overline{x}_k}^{\infty}  e^{-r_k^\ast s}  \mathrm{d}s \right]	 < \infty$\bsc for some positive constant $C$. Therefore, $\Omega^{J \ast}$ is an admissible control policy and we have  $V \left(\mathbf{Q}  \right)=J\left(\mathbf{Q} \right)$ and $\theta=c^\infty$. Furthermore, using Corollary \ref{cor1}, we have $V^\ast\left(\mathbf{Q}  \right)=J\left(\mathbf{Q} \right)+o(1)$ and $\theta^\ast=c^\infty+o(1)$ for sufficiently small $\tau$.

\section*{Appendix D: Proof of Lemma \ref{linearAp}}
\subsubsection{Proof of the decomposable structure of the base PDE}In the base PDE, since $L_{kj}=0$ for all $k,j, k \neq j$, the associated PDE becomes:
\begin{small}\begin{align}	
		& \mathbb{E}\big[ \min_{\mathbf{F}, \mathbf{U}} \big[\sum_{k=1}^K \big(\mathrm{Tr}\left(  \mathbf{F}_k  \mathbf{F}_k^\dagger\right)+\gamma_k e^{-\eta \left[Q_k - Q^l \right]^+}+\beta_k e^{-\eta \left[Q^h - Q_k \right]^+} \notag \\
		&  +  \frac{\partial  J \left(\mathbf{Q} ;{0}\right)} {\partial Q_k} \big(R_k^0\left(\mathbf{H}, \mathbf{F}_k, \mathbf{U}_k \right)  -\mu_k\big) \big)\big]\big| \mathbf{Q}  \big] - c^{\infty}=0	\label{basepde}
\end{align}\end{small}with boundary condition $J_k(Q_k^\star)=0$, for some $Q_k^\ast$, where we denote $R_k^0\left(\mathbf{H}, \mathbf{F}_k, \mathbf{U}_k \right) =W \log_2 \det ( \mathbf{I}  +L_{kk}\mathbf{U}_k^\dagger \mathbf{H}_{kk}\mathbf{F}_k \mathbf{F}_k^\dagger \mathbf{H}_{kk}^\dagger \mathbf{U}_k ) $. We have the following lemma establishing the decomposable structure of the $J \left( \mathbf{Q} ; {0} \right)$  and $c^{\infty}$ in (\ref{basepde}).
\begin{Lemma}	[Decomposed Optimilaty Equation]	\label{decomplem}
	Suppose there exist $c_k^\infty$ and $J_k \left(Q_k \right) \in \mathcal{C}^2\left(\mathbb{R}_+ \right)$ that solve the following per-flow PDE:
	\bs\begin{align}
		&\mathbb{E}\big[\min_{\mathbf{F}_k} \big[ \mathrm{Tr}(  \mathbf{F}_k  \mathbf{F}_k^\dagger)+\gamma_k e^{-\eta \left[Q_k - Q^l \right]^+}+\beta_k e^{-\eta \left[Q^h - Q_k \right]^+}  \notag \\ 
		&	 + J_k'(Q_k)\left(R_k^0\left(\mathbf{H}, \mathbf{F}_k, \mathbf{U}_k \right)  -\mu_k\right)\big] \big| \mathbf{Q}  \big] - c^{\infty}_k=0	\label{perflowhjb}
	\end{align}\bsc
	Then, $ J\left( \mathbf{Q};{0}\right)=\sum_{k=1}^K J_k\left(Q_k \right)$ and $c^\infty=\sum_{k=1}^K c_k^\infty$ satisfy  (\ref{basepde}).~\hfill\IEEEQED
\end{Lemma}

Lemma \ref{decomplem} can be proved using the fact that the dynamics of the playback buffer  are decoupled when $L=0$. The details are omitted for conciseness.

\subsubsection{Solving the per-flow PDE} We first write $\mathbf{F}_k=\widetilde{\mathbf{F}}_k\boldsymbol{\Sigma}_k$, where $\widetilde{\mathbf{F}}_k=\left[\mathbf{f}_{k1}, \dots, \mathbf{f}_{kd} \right] \in \mathbb{C}^{N_t \times d}$ with $\|\mathbf{f}_{ki}\|=1$ ($\forall i=1,\dots, d$), and $\boldsymbol{\Sigma}_k=\mathrm{diag}\left({p_{k1},\dots,p_{kd}}\right)$ where $p_{ki}$ is the power allocated for the $i$-th data stream.  Let the singular value decomposition of the channel matrix be $\mathbf{H}_{kk}=\mathbf{M}_k \boldsymbol{\Lambda}_k \mathbf{N}_k^\dagger$, where $\mathbf{M}_k  \in \mathbb{C}^{N_r \times N_r}$ and $\mathbf{N}_k \in \mathbb{C}^{N_t \times N_t}$ are unitary matrices and $ \boldsymbol{\Lambda}_k \in \mathbb{R}^{N_r \times N_t}$ whose diagonal elements $\sigma_{k1}\geq \cdots \geq \sigma_{kd}$ are the  singular values of $\mathbf{H}_{kk}$ and the off-diagonal elements are zero. Therefore, the problem in the base PDE (\ref{perflowhjb}) becomes:
\bs\begin{align}
	&\min_{p_{k1},\dots, p_{kd}, \atop \widetilde{\mathbf{F}}_k, \mathbf{U}_k} \sum_{i=1}^d p_{ki} + J_k'(Q_k) W \log_2 \det \left( \mathbf{I} \notag \right.\\
	&\left. +L_{kk} \mathbf{U}_k^\dagger \mathbf{M}_k \boldsymbol{\Lambda}_k \mathbf{N}_k^\dagger\widetilde{\mathbf{F}}_k\boldsymbol{\Sigma}_k \boldsymbol{\Sigma}_k^\dagger \widetilde{\mathbf{F}}_k^\dagger \mathbf{N}_k \boldsymbol{\Lambda}_k^\dagger  \mathbf{M}_k^\dagger \mathbf{U}_k \right)
\end{align}\bsc
The above problem is the classical MIMO beamforming control problem \cite{tse} and the optimal  $\widetilde{\mathbf{F}}_k^\ast$ is the first $d$ columns of $\mathbf{N}_k$, the optimal $\widetilde{\mathbf{U}}_k^\ast$ is the first $d$ columns of $\mathbf{M}_k$, and the optimal power allocation is given by
\bs\begin{align}	\label{magniopt}
	p_{ki}^{\ast}(\sigma_{ki})=\left( -\frac{J_k'\left(Q_k \right)W }{ \ln 2} - \frac{1}{L_{kk} \sigma_{ki}^2 }\right)^+
\end{align}  \bsc
We next calculate the expectations involved in (\ref{perflowhjb}). Specifically, substituting the optimal precoding matrix  $\widetilde{\mathbf{F}}_k^\ast$, $\mathbf{U}_k^\ast$ and power $p_{ki}^\ast$ into (\ref{perflowhjb}), we obtain that
\begin{small}\begin{align}
	&\mathbb{E}\left[ \mathrm{Tr}\left(  \mathbf{F}_k^\ast  (\mathbf{F}_k^\ast)^\dagger\right)\right]=\mathbb{E}\big[ \sum_{i=1}^d p_{ki}^{\ast}(\sigma_{ki})\big]=d \mathbb{E}\left[ p_{k1}^{\ast}(\sigma_{k1})\right]	\label{exp1}\\
	&\mathbb{E}\left[ R_k^0\left(\mathbf{H}, \mathbf{F}_k^\ast, \mathbf{U}_k^\ast \right) \right]=\mathbb{E}\big[W \sum_{i=1}^d \log_2\left(1+L_{kk} \sigma_{ki}^2 p_{ki}^\ast(\sigma_{ki})\right)\big]\notag \\
	&=d \mathbb{E}\left[ W \log_2\left(1+L_{kk} \sigma_{k1}^2 p_{k1}^\ast(\sigma_{k1})\right)\right]	\label{exp2}
\end{align}\end{small}which depend on the distribution of one of the unordered singular values. Let $b=\max\{N_t,N_r \}$. According to \cite{mimoking}, the distribution of any $\sigma_{ki}^2$ is given by: $f_{\sigma_{ki}^2}(x)=\frac{x^{d-b} e^{x} }{d}\sum_{n=1}^d \varphi_n^2(x)$, where $\varphi_{n+1}(x)$ is given by $\varphi_{n+1}(x) =\left[\frac{1}{n!(n+b-d)!}\right]^{1/2}\frac{\mathrm{d}^{n}}{\mathrm{d} x^{n}}(e^{-x} x^{n+b-d})=\left[\frac{1}{n!(n+b-d)!}\right]^{1/2} \sum_{l=0}^n A_{n,l} (-1)^l C_n^l e^{-x} x^{b-d+l}$, with $A_{n,l}=1$ if $l=n$ and $A_{n,l}=\prod_{r=0}^{n-l-1} (n-r+b-d)$ if $l<n$, and $C_n^l$ is the Binomial coefficient. Therefore, the distribution of any $\sigma_{ki}^2$  can be rewritten as $f_{\sigma_{ki}^2}(x)=\frac{x^{b-d} e^{-x} }{d}\sum_{n=0}^{d-1} \frac{1}{n!(n+b-d)!}\left[ \sum_{l=0}^n A_{n,l} (-1)^l C_n^l  x^{l} \right]^2	=\frac{x^{b-d} e^{-x} }{d}\sum_{n=0}^{d-1} b_n\left[ \sum_{l=0}^n a_{n,l,l}  x^{2l} + \sum_{l=0}^n\sum_{j>l}2 a_{n,l,j} x^{l+j}\right]$, where we denote $b_n=\frac{1}{n!(n+b-d)!}$,  and $a_{n,l,j}=(A_{n,l} C_n^l)^2$ when $l=j$ and $a_{n,l,j}= A_{n,l}  A_{n,j} (-1)^{l+j}C_n^l C_n^j $ when $j>l$. We further denote $s \triangleq b-d$ and $t_k \triangleq \frac{-\ln 2}{W L_{kk}}$. We then calculate (\ref{exp1}) and (\ref{exp2}) as follows:
\begin{small}\begin{align}
	& d \mathbb{E}\left[ p_{k1}^{\ast}(\sigma_{k1})\right]=  \frac{1 }{ L_{kk} } \sum_{n=0}^{d-1} b_n\left[ \sum_{l=0}^n a_{n,l,l} \left[ \frac{J_k'\left(Q_k \right) }{t_k  } \right.\right.  \label{56equ1} \\
	& \hspace{-0.6cm}\left.\left.  G\left(1+2l+s,\frac{t_k}{J_k'\left(Q_k \right)} \right) -  G\left(2l+s,\frac{t_k}{J_k'\left(Q_k \right)} \right)\right]+ \sum_{l=0}^n\sum_{j>l}2 a_{n,l,j} \Big[ \right. 	\notag  \\
	 &\left.\frac{J_k'\left(Q_k \right) }{t_k  }   G\left(1+l+j+s,\frac{t_k}{J_k'\left(Q_k \right)} \right) - G\left(l+j+s,\frac{t_k}{J_k'\left(Q_k \right)} \right)\Big]\right]	\notag \\
	&d \mathbb{E}\left[ W \log_2\left(1+L_{kk} \sigma_{k1}^2 p_{k1}^\ast(\sigma_{k1})\right)\right]=\frac{W}{\ln 2} \sum_{n=0}^{d-1} b_n\bigg[ \sum_{l=0}^n a_{n,l,l}   \notag \\
	& M\left(\{1,1\},\{0,0,1+2l+s\},\frac{t_k}{J_k'\left(Q_k \right)} \right)   	\label{56equ}\\
	&+ \sum_{l=0}^n\sum_{j>l}2 a_{n,l,j} M\left(\{1,1\},\{0,0,1+l+j+s\},\frac{t_k}{J_k'\left(Q_k \right)} \right) \bigg]\notag
\end{align}\end{small}
where\footnote{$M_g\left(\{a_{1},\dots,a_n\},\{b1,\dots,b_m\},  z\right)=\frac{1}{2\pi i}\int_{\mathcal{L}}\frac{\prod_{k=1}^m \Gamma(b_k-s)}{\prod_{k=1}^n \Gamma(a_k-s)}z^{s} \mathrm{d}s$ is the Meijer G-function, where $\Gamma(a)=G_a(a,0)$. $G_a\left(a,x \right)=\int_{x}^{\infty} t^{a-1} e^{-t} \mathrm{d}t$ is the incomplete gamma function.} $G(m,x)=G_a(m,x)$ (Gamma function) if $x> 0$ and equals to zero otherwise. $M( \{  \},\{ \}  , x)=M_g( \{  \},\{ \}  , x)$ if $x>0$ and  equals to zero otherwise. We then calculate $c_k^\infty$. Assuming $e^{\eta \left(Q^l-Q^h  \right)} < \frac{\gamma_k}{\beta_k} < e^{\eta \left(Q^h-Q^l  \right)}$, and we define the following \emph{target} operating queue  regime $Q_k^\star$ (achieving the minimum of the per-stage cost function $c_k\left(Q_k,\mathbf{F}_k\right)$ within the domain $(Q^l, Q^h )$ for given precoding matrix $\mathbf{F}_k$:
\begin{small}\begin{align}	
	Q_k^\star  = \min_{Q_k} c_k\left(Q_k,\mathbf{F}_k\right)	 = \frac{Q^l+Q^h}{2} + \frac{1}{2 \eta} \ln \frac{\gamma_k}{\beta_k} \in (Q^l, Q^h )
\end{align}\end{small}

To satisfy  boundary condition $J_k(Q_k^\star)=0$, we require that $c_k^\infty =\mathrm{R.H.S. of \ (\ref{56equ1})}\big|_{Q_k=Q_k^\star}+\gamma_k e^{-\eta \left[Q_k^\star - Q^l \right]^+}+\beta_k e^{-\eta \left[Q^h - Q_k^\star \right]^+}$ and $\mathbb{E}\left[ R_k^0\left(\mathbf{H}, \mathbf{F}_k^\ast, \mathbf{U}_k^\ast \right) \right]\big|_{Q_k=Q_k^\star}= \mathrm{R.H.S. of \ (\ref{56equ})}\big|_{Q_k=Q_k^\star}=\mu_k$. Therefore, we have
\begin{small}\begin{align}
	 &\hspace{-0.7cm} c_k^\infty=\frac{1 }{ L_{kk} } \sum_{n=0}^{d-1} b_n\left[ \sum_{l=0}^n a_{n,l,l} \left[\frac{\lambda_k}{t_k} G\left(1+2l+s,\frac{t_k}{\lambda_k}\right) -  G\left(2l+s,\frac{t_k}{\lambda_k}\right)\right]\right.  \notag \\
	& \left. \hspace{-0.5cm} + \sum_{l=0}^n\sum_{j>l}2 a_{n,l,j} \left[\frac{\lambda_k}{t_k} G\left(1+l+j+s,\frac{t_k}{\lambda_k}\right) - G\left(l+j+s,\frac{t_k}{\lambda_k}\right)\right]\right] \notag \\
	&+\gamma_k e^{-\eta \left[Q_k^\star - Q^l \right]}+\beta_k e^{-\eta \left[Q^h - Q_k^\star \right]}\label{cinfequ}
\end{align}\end{small}where $\lambda_k \in \mathbb{R}_-$ satisfies $\frac{W}{\ln 2} \sum_{n=0}^{d-1} b_n[ \sum_{l=0}^n a_{n,l,l}  M(\{1,1\},\{0,0,1+2l+s\},\frac{t_k}{\lambda_k})  + \sum_{l=0}^n\sum_{j>l}2 a_{n,l,j}  M( \{1,1\}, \{0,0,1+l+j+s\},\frac{t_k}{\lambda_k}) ]=\mu_k$. Substituting the results on the expectations in (\ref{56equ1}) and (\ref{56equ}), and the result on $c_k^\infty$ in (\ref{cinfequ}) into (\ref{perflowhjb}), we can obtain the following  fixed point  equation determining $J_k'(Q_k)$:
\begin{align}
		g(Q_k, J_k')=0	\label{fixedequ}
\end{align}
where we denote \begin{small}$g(Q_k, J_k')\triangleq \frac{1 }{ L_{kk} } \sum_{n=0}^{d-1} b_n\left[ \sum_{l=0}^n a_{n,l,l} \left[ \frac{J_k' }{t_k  }  G(1+2l+s,\frac{t_k}{J_k'} ) -  G(2l+s,\frac{t_k}{J_k'} )\right ]  \right. \\ \left.+    \sum_{l=0}^n\sum_{j>l} 2 a_{n,l,j} \left[\frac{J_k' }{t_k  }   G(1+l+j+s,\frac{t_k}{J_k'} ) - G(l+j+s,\frac{t_k}{J_k'} )\right]\right] \\ + \gamma_k e^{-\eta \left[Q_k - Q^l \right]^+}+\beta_k e^{-\eta \left[Q^h - Q_k \right]^+} +J_k'(\frac{W}{\ln 2} \sum_{n=0}^{d-1}  \left[ \sum_{l=0}^n a_{n,l,l} \right. \\ \left.  M(\{1,1\},\{0,0,1+2l+s\},\frac{t_k}{J_k'} ) + \sum_{l=0}^n\sum_{j>l}2 a_{n,l,j}  M(\{1,1\}, \right. \\ \left.   \{0,0,1+l+j+s\},\frac{t_k}{J_k'} ) \right] -\mu_k) - c^{\infty}_k$\end{small}. It can be shown that for fixed $Q_k \in [0, Q_k^\ast]$, $g(Q_k, J_k')$ is strictly increasing w.r.t. $J_k'$ over $(-\infty, \lambda_k]$ and $g(Q_k, \lambda_k)\geq 0$. Then, it follows that $g(Q_k, J_k')=0$ has a unique solution over $J_k' \in (-\infty, \lambda_k]$. Similarly, it can be shown that for fixed $Q_k \in [Q_k^\ast, \infty)$, $g(Q_k, J_k')$ is strictly decreasing w.r.t. $J_k'$ over $[\lambda_k, \infty)$ and $g(Q_k, \lambda_k)\geq 0$. Then, it follows that $g(Q_k, J_k')=0$ has a unique solution over $J_k' \in [\lambda_k, \infty)$. 
\subsubsection{Asymptotic property of $J_k\left(Q_k \right)$} based on the above analysis on the behavior of $g(Q_k, J_k')$, we have that for sufficiently large $Q_k$, $J_k'(Q_k)$ become positive and the fixed point equation is simplified as follows: $\beta_k - J_k'(\infty)\mu_k - c_k^\infty=0$. Assuming $\beta_k > c_k^\infty$, then $J_k'(\infty)=\frac{\beta_k - c_k^\infty}{\mu_k}>0$. Denote $C_k=\frac{\beta_k - c_k^\infty}{\mu_k}$. Thus,  we have that  $J_k(Q_k)=C_k Q_k, \ \text{as } Q_k \rightarrow \infty$.

\section*{Appendix E: Proof of Theorem \ref{ErrorEg2}}
Taking the first order Taylor expansion of the L.H.S. of the PDE in (\ref{cenHJB}) at $L_{kj}=0$ ($\forall k,j$), $\mathbf{F}_k=\mathbf{F}_k^{\ast}$ and $\mathbf{U}_k=\mathbf{U}_k^{\ast}$ (where $\mathbf{F}_k^\ast$ and $\mathbf{U}_k^{\ast}$ are the optimal control actions solving the per-flow PDE  in  (\ref{perflowhjb})), and using  parametric optimization analysis \cite{paraanay},  we have the following result regarding the approximation error:
\bs\begin{align}
	J\left(\mathbf{Q} ; {L} \right)-J\left(\mathbf{Q} ; 0 \right) =\sum_{k=1}^K \sum_{j \neq k} L_{kj}  \widetilde{J}_{kj}(\mathbf{Q} )+ \mathcal{O}(L^2)	\label{tayloee}
\end{align}\bsc where \bs$\widetilde{J}_{kj}(\mathbf{Q} )$\bsc is meant to capture the coupling terms in \bs$J\left(\mathbf{Q} ; {L}\right)$\bsc  which satisfies the following PDE:
\begin{small}\begin{align}	
	&\sum_{i=1}^K \left(\mathbb{E} \left[  R_i^0\left(\mathbf{H}_{ii}, \mathbf{F}_i^{\ast},\mathbf{U}_i^{\ast}\right) | {Q}_i   \right]  -\mu_i\right) \frac{\partial \widetilde{J}_{kj}\left(\mathbf{Q} \right) }{\partial Q_i}  \notag \\
	&+ \mathbb{E} \left[J_k'(Q_k) \frac{\partial R_k\left(\mathbf{H}, \mathbf{F}^{\ast}, \mathbf{U}_k^{\ast}\right)}{\partial L_{kj}}\bigg|_{L=0}\bigg| \mathbf{Q}   \right] =0	\label{PDEappd}
\end{align}\end{small}
with boundary condition \bs$\widetilde{J}_{kj}\left(\mathbf{Q}  \right)\big|_{Q_i = Q_k^\star}=0$ or $\widetilde{J}_{kj}\left(\mathbf{Q}  \right)\big|_{Q_j =Q_j^\star}=0$\bsc.  We next calculate the two expectations involved in the above equation. According to the analysis of the fixed point equation in (\ref{fixedequ}) in Appendix D.2, $J_k'(Q_k)$ decreases and approaches to $-\infty$ as $Q_k$ decreases on the  domain $(-\infty, Q_k^\star]$, while $J_k'(Q_k)$ increases and approaches to $\infty$ as $Q_k$ increases on the  domain $[Q_k^\star, \infty)$. Therefore, we calculate the expectations in (\ref{PDEappd}) by taking  into account of the queue regions.  For the first expectation, according to (\ref{asympoticJ_k}) and (\ref{exp2}), we have that if $Q_i\geq Q_i^\star$, as $Q_i$ goes to infinity, then
\bs \begin{align}	\label{fin1}
	\mathbb{E} \left[  R_i^0\left(\mathbf{H}_{ii}, \mathbf{F}_i^{\ast},  \mathbf{U}_i^{\ast}\right) | Q_i   \right]=0, \quad \text{for large } Q_i
\end{align}\bsc
since the water level in (\ref{magniopt}) which is determined by $-J_k'(Q_k)$ becomes negative (as $Q_i\rightarrow \infty$).  If $Q_i<Q_i^\star$, $J_k'(Q_k)$ approaches $-\infty$ as $Q_i$ decreases. Based on the  asymptotic behavior of the  Gamma  function and the Meijer-G function: \bs$G(n,\frac{a}{y})=(n-1)!+o(1)$,  $M\big(\{1,1\},\{0,0,N\},\frac{1}{ a y} \big)=(N-1)!\ln a y + (N-1)!P_g^0(N)+o(1)$\bsc  for $a$, $y$ with the same sign, where $P_g^0(x)$ is the Polygamma function.  From (\ref{fixedequ}), we have
\bs \begin{align}
	&-c^1_k J_k'(Q_k)-c_k^2+\beta_k + J_k'(Q_k) \left[ c^1_k \ln(-J_k'(Q_k))+c^3_k-\mu_k\right]	\notag \\
	&=c_k^\infty \label{pareJkEg22}
	\end{align} \bsc
where we denote $c^1_k$, $c^2_k$, $c^3_k$ as follows\footnote{$P_g^0(x)=\left(\log\left(\Gamma(x) \right)\right)'$ is the polygamma function.}:
\begin{small}\begin{align}
	&c^1_k\triangleq 	\label{111equ} \\
	&\frac{1 }{ L_{kk} } \sum_{n=0}^{d-1} b_n\Bigg[ \sum_{l=0}^n a_{n,l,l}  \frac{\left(2l+s \right)!  }{-t_k  }  + \sum_{l=0}^n\sum_{j>l}2 a_{n,l,j}\left[\frac{  \left(l+j+s \right)! }{-t_k  } \right]\Bigg] \notag \\
	&c^2_k \triangleq  \label{222equ}	\\
	&\hspace{-0.3cm} \frac{1 }{ L_{kk} } \sum_{n=0}^{d-1} b_n\Bigg[ \sum_{l=0}^n a_{n,l,l}  \left(2l+s-1\right)!+ \sum_{l=0}^n\sum_{j>l}2 a_{n,l,j}\left(l+j+s-1 \right)!\Bigg]  \notag\\
	&c^3_k\triangleq \frac{W}{\ln 2} \sum_{n=0}^{d-1} b_n\bigg( \sum_{l=0}^n a_{n,l,l}  (2l+s)! \left[-\ln(-t_k)+P_g^0(1+2l+s) \right] \notag \\
	 & + \sum_{l=0}^n\sum_{j>l}2 a_{n,l,j}(l+j+s)!\left[-\ln(-t_k)+P_g^0(1+l+j+s)\right]\bigg) 	\label{333equ} 	
\end{align}\end{small} where $c_k^1$ and $c_k^2$ are positive. Therefore, for sufficiently small $Q_k$, we can rewrite (\ref{pareJkEg22}) as 
\begin{align}	\label{fixedpointws}
	\hspace{-0.5cm}-c^1_k D_k-c_k^2+\beta_k + D_k \left[ c^1_k \ln(-D_k)+c^3_k-\mu_k\right]=c_k^\infty
\end{align}   
where we use $D_k$ to represent $J_k'(Q_k)$ for large $Q_k$.  Note that (\ref{fixedpointws}) has a unique solution if  $\beta_k>c_k^\infty$. Therefore, for the first expectation in (\ref{PDEappd}), for small  $Q_i$, we have
\bs \begin{align}	\label{fin2}
	\mathbb{E} \left[  R_i^0\left(\mathbf{H}_{ii}, \mathbf{F}_i^{\ast},  \mathbf{U}_i^{\ast}\right) | Q_i   \right]=c_i^1 \ln(-D_i)+c_i^3, \quad \text{for small } Q_i
\end{align} \bsc
Furthermore, $c_i^1 \ln(-D_i)+c_i^3-\mu_i>0$ according to (\ref{fixedpointws}). For the second expectation in (\ref{PDEappd}), we have  \begin{small}$\frac{\partial R_k\left(\mathbf{H}, \mathbf{F}^{\ast},  \mathbf{U}_k^{\ast}\right)}{\partial L_{kj}}\bigg|_{L=0}=-\mathrm{Tr}( (\mathbf{I}+ L_{kk}(\mathbf{U}_k^\ast)^\dagger\mathbf{H}_{kk}\mathbf{F}_k^\ast (\mathbf{F}_k^\ast)^\dagger \mathbf{H}_{kk}^\dagger\mathbf{U}_k^\ast )^{-1}   L_{kk}(\mathbf{U}_k^\ast)^\dagger\mathbf{H}_{kk}\mathbf{F}_k^\ast (\mathbf{F}_k^\ast)^\dagger \mathbf{H}_{kk}^\dagger  \mathbf{H}_{kj}\\ \mathbf{F}_j^\ast (\mathbf{F}_j^\ast)^\dagger \mathbf{H}_{kj}^\dagger \mathbf{U}_k^\ast)$. \end{small}Substituting $\mathbf{F}_k^\ast$ and $\mathbf{U}_k^{\ast}$   in  (\ref{perflowhjb}), we obtain
\begin{small}\begin{align}
	&\mathbb{E}\left[\frac{\partial R_k\left(\mathbf{H}, \mathbf{F}^{\ast}, \mathbf{U}_k^\ast\right)}{\partial L_{kj}}\bigg|_{L=0}\bigg|\mathbf{Q}\right]	\notag \\
	&=-\mathbb{E}\left[\sum_{n=1}^d \frac{L_{kk}p_{kn}^\ast(\sigma_n)\sigma_n^2}{L_{kk}p_{kn}^\ast(\sigma_n)\sigma_n^2+1}p_{jn}^\ast(\sigma_n)\bigg|\mathbf{Q}\right]	\label{finalequllll}
\end{align}\end{small}Similarly, if either $Q_k$ or $Q_j$ is sufficiently large, (\ref{finalequllll}) equals to zero. Otherwise, (\ref{finalequllll}) equals to $\frac{-1}{d}\left( c_k^1 D_k +c_k^2 \right)  \cdot\frac{\ln 2}{D_k W}\left(c_j^1D_j+c_j^2 \right) \triangleq G_{kj}$. Denote $G_k \triangleq c_k^1 \ln(-D_k)+c_k^3 - \mu_k$. Combining (\ref{fin1}) and (\ref{fin2}), according to Section \emph{3.8.1.2}  of \cite{handbookPDE} and taking into account of the boundary conditions, by solving (\ref{PDEappd}) we have \bs$\widetilde{J}_{kj}\left(\mathbf{Q} \right) =o(1)$\bsc  if either $Q_k\geq Q_k^\star$ or $Q_j\geq Q_j^\star$, and \bs$\widetilde{J}_{kj}\left(\mathbf{Q} \right) =-\frac{D_kG_{kj}}{2{G}_k} (Q_k-Q_k^\star)-\frac{D_k{G}_{kj}}{2{G}_j} (Q_j-Q_j^\star)+o(Q_k)+o(Q_j)$\bsc otherwise. Substituting it into (\ref{tayloee}) and denoting $E_{kj}=\frac{D_kG_{kj}}{2{G}_k} $, we obtain the approximation error in Theorem \ref{ErrorEg2}.

\section*{Appendix F: Proof of Lemma \ref{globalopt}}
1) \emph{Convergence property:} The proof follows similar approach as in \cite{WMMSE} by showing (\ref{utility}) and (\ref{utility1}) have the same KKT conditions at the stationary point $\left\{\mathbf{F}(\infty), \mathbf{Z}(\infty), \mathbf{K}(\infty) \right\}$. Details are omitted due to page limit.  	

2) \emph{Asymptotically optimality:} we next prove the asymptotically property of Algorithm \ref{wmmsealg}. Denote the objective function in (\ref{utility}) as $ f\left( \mathbf{F}, L \right)$. We have the following lemma on the convexity for $f\left( \mathbf{F}, L \right)$.
\begin{Lemma}	[Convexity of $f\left( \mathbf{F}, L \right)$ for Sufficiently Small $L$]	\label{lemma8}
	$f\left( \mathbf{F}, L \right)$ is a convex function of $\mathbf{F}=\left\{\mathbf{F}_k:  \right.  \\ \left. \forall k \right\}$ when $L$ is sufficiently small.~\hfill\IEEEQED
\end{Lemma}
\begin{proof}
	According to \cite{boyd}, we have the following argument regarding the convexity of a function $f(\mathbf{x})$: given any two different feasible points $\mathbf{x}_1$ and $\mathbf{x}_2$, define $g(t)= f(t \mathbf{x}_1+(1-t)\mathbf{x}_2)$, $0\leq t \leq 1$, then $f(\mathbf{x})$ is a convex function of $\mathbf{x}$ if and only if $g(t)$ is a convex function of $t$, which is equivalent to $\frac{\mathrm{d}^2 g(t)}{\mathrm{d}t^2} \geq 0$ for $0 \leq t \leq 1$.
	
	Therefore, we consider the convex combination of two different feasible solutions $\mathbf{F}^{(1)}=\{\mathbf{F}_k^{(1)}:  \forall k \}$ and $\mathbf{F}^{(2)}=\{\mathbf{F}_k^{(2)}:  \forall k\}$ as follows: $\mathbf{F}^c=\{\mathbf{F}_k^c=t \mathbf{F}_k^{(1)} + (1-t) \mathbf{F}_k^{(2)} :  \forall k \}$ and $0 \leq t \leq 1$. Denote $\mathbf{F}_{-k}=\{\mathbf{F}_j: \forall j \neq k\}$, $\mathbf{G}_k(\mathbf{F}_{-k}) = \mathbf{I}+\sum_{j\neq k}L_{kj}\mathbf{H}_{kj}\mathbf{F}_j\mathbf{F}_j^\dagger \mathbf{H}_{kj}^\dagger $, $\mathbf{Y}_k=\mathbf{F}_k^{(1)}-\mathbf{F}_k^{(2)}$ and $a_k=\frac{W}{\ln 2}\frac{\partial \widetilde{V}\left(\mathbf{Q}\right)}{\partial Q_k}$. W.l.o.g, we assume $a_k\leq 0$ for all $k$ (since for $a_k>0$, the associated optimal $\mathbf{F}_k=\mathbf{0}$, and thus we can focus on those Tx-Rx pair $j$ such that $a_j\leq0$. See Lemma \ref{lemmaappr} in Appendix C for the detailed proof), then the second order derivative of $f\left( \mathbf{F}^c, L \right)$ is: \bs$\frac{\mathrm{d}^2 f( \mathbf{F}^c, L )}{\mathrm{d}t^2}	=\sum_{k } \mathrm{Tr}( \mathbf{Y}_k \mathbf{Y}_k^{\dagger } + \mathbf{Y}_k \mathbf{Y}_k ^\dagger  -a_k(( \mathbf{G}_k(\mathbf{F}_{-k}^c)  + L_{kk} \mathbf{H}_{kk} \mathbf{F}_k^c \mathbf{F}_k^{c\dagger} \mathbf{H}_{kk}^\dagger)^{-1} ( \frac{\mathrm{d}\mathbf{G}_k(\mathbf{F}_{-k}^c)}{\mathrm{d}t}  + L_{kk} \mathbf{H}_{kk}\mathbf{Y}_k\mathbf{Y}_k^\dagger \mathbf{H}_{kk}^\dagger ) ( \frac{\mathrm{d}\mathbf{G}_k(\mathbf{F}_{-k}^c)}{\mathrm{d}t}  + L_{kk} \mathbf{H}_{kk}\mathbf{Y}_k\mathbf{Y}_k^\dagger \mathbf{H}_{kk}^\dagger )^{-1} ( \mathbf{G}_k(\mathbf{F}_{-k}^c)  + L_{kk} \mathbf{H}_{kk} \mathbf{F}_k^c \mathbf{F}_k^{c\dagger} \mathbf{H}_{kk}^\dagger)+ \mathbf{G}_k^{-1}(\mathbf{F}_{-k}^c) (\frac{\mathrm{d}\mathbf{G}_k(\mathbf{F}_{-k}^c)}{\mathrm{d}t}) \mathbf{G}_k^{-1}(\mathbf{F}_{-k}^c) (\frac{\mathrm{d}\mathbf{G}_k(\mathbf{F}_{-k}^c)}{\mathrm{d}t})))	$\bsc, where $ \frac{\mathrm{d}\mathbf{G}_k(\mathbf{F}_{-k}^c) }{\mathrm{d}t}=\sum_{j\neq k}L_{kj}\mathbf{H}_{kj}\mathbf{Y}_j \mathbf{Y}_j^\dagger \mathbf{H}_{kj}^\dagger$ does not depend on $t$. As $L$ becomes sufficiently small,  $\frac{\mathrm{d}\mathbf{G}_k(\mathbf{F}_{-k}^c)}{\mathrm{d}t} $ is proportional to $L$ and $\frac{\mathrm{d}\mathbf{G}_k(\mathbf{F}_{-k}^c)}{\mathrm{d}t}  +  L_{kk} \mathbf{H}_{kk}\mathbf{Y}_k\mathbf{Y}_k^\dagger \mathbf{H}_{kk}^\dagger$ is dominated by $ L_{kk} \mathbf{H}_{kk}\mathbf{Y}_k\mathbf{Y}_k^\dagger \mathbf{H}_{kk}^\dagger$.  $\mathbf{G}_k^{-1}(\mathbf{F}_{-k}^c) \left(\frac{\mathrm{d}\mathbf{G}_k(\mathbf{F}_{-k}^c)}{\mathrm{d}t}\right)  \mathbf{G}_k^{-1}(\mathbf{F}_{-k}^c) \left(\frac{\mathrm{d}\mathbf{G}_k(\mathbf{F}_{-k}^c)}{\mathrm{d}t}\right)$ is proportional to $L^2$ and hence it has little impact on the first term in the derivative and can be ignored. Therefore, 
	\begin{small}\begin{align}
		 &\frac{\mathrm{d}^2 f\left( \mathbf{F}^c, L \right)}{\mathrm{d}t^2} 	\approx  \sum_{k } \mathrm{Tr}\left( \mathbf{Y}_k \mathbf{Y}_k^{\dagger } + \mathbf{Y}_k \mathbf{Y}_k ^\dagger -a_k\left( \left( \mathbf{G}_k(\mathbf{F}_{-k}^c)  \right. \right.\right. \notag \\
		 & \left. \left.\left. + L_{kk} \mathbf{H}_{kk} \mathbf{F}_k^c \mathbf{F}_k^{c\dagger} \mathbf{H}_{kk}^\dagger\right)^{-1} \left( \frac{\mathrm{d}\mathbf{G}_k(\mathbf{F}_{-k}^c)}{\mathrm{d}t}  + L_{kk} \mathbf{H}_{kk}\mathbf{Y}_k\mathbf{Y}_k^\dagger \mathbf{H}_{kk}^\dagger \right) \right.\right. \notag \\
		& \left.\left. \left( \frac{\mathrm{d}\mathbf{G}_k(\mathbf{F}_{-k}^c)}{\mathrm{d}t}  + L_{kk} \mathbf{H}_{kk}\mathbf{Y}_k\mathbf{Y}_k^\dagger \mathbf{H}_{kk}^\dagger \right)^{-1}\right.\right.  \notag \\
		&\left.\left. \left( \mathbf{G}_k(\mathbf{F}_{-k}^c)  + L_{kk} \mathbf{H}_{kk} \mathbf{F}_k^c \mathbf{F}_k^{c\dagger} \mathbf{H}_{kk}^\dagger\right)\right)\right)	 \label{68equ}
	\end{align}\end{small}Denote \bs$\mathbf{A}_k=( \mathbf{G}_k(\mathbf{F}_{-k}^c)  + L_{kk} \mathbf{H}_{kk} \mathbf{F}_k^c \mathbf{F}_k^{c\dagger} \mathbf{H}_{kk}^\dagger)^{-1}$ and $\mathbf{B}_k=\frac{\mathrm{d}\mathbf{G}_k(\mathbf{F}_{-k}^c)}{\mathrm{d}t}  + L_{kk} \mathbf{H}_{kk}\mathbf{Y}_k\mathbf{Y}_k^\dagger \mathbf{H}_{kk}^\dagger $\bsc. Since $\mathbf{A}_k$ is positive semidefinite, there exists a matrix $\mathbf{X}_k$ such that $\mathbf{A}_k=\mathbf{X}_k\mathbf{X}_k^\dagger$. Thus, from (\ref{68equ}), we have \bs$\frac{\mathrm{d}^2 f( \mathbf{F}^c, L )}{\mathrm{d}t^2} \approx  \sum_{k } \mathrm{Tr}( 2\mathbf{Y}_k \mathbf{Y}_k^{\dagger }   -a_k \mathbf{A}_k\mathbf{B}_k \mathbf{A}_k\mathbf{B}_k)=\sum_{k } \mathrm{Tr}( 2\mathbf{Y}_k \mathbf{Y}_k^{\dagger }   -a_k\mathbf{X}_k^\dagger\mathbf{B}_k \mathbf{X}_k\mathbf{X}_k^\dagger\mathbf{B}_k\mathbf{X}_k)	=\sum_{k } \mathrm{Tr}( 2\mathbf{Y}_k \mathbf{Y}_k^{\dagger }   -a_k(\mathbf{X}_k^\dagger\mathbf{B}_k \mathbf{X}_k)(\mathbf{X}_k^\dagger\mathbf{B}_k \mathbf{X}_k)^\dagger)	\geq 0$\bsc, for sufficiently small $L$, where the last equality is due to the fact that $a_k\leq 0$ and $\mathbf{B}_k$ is Hermitian. Therefore, $f\left( \mathbf{F}, L \right)$ is convex for sufficiently small $L$.
\end{proof}

	Based on Lemma \ref{lemma8}, for sufficiently small   $L$, the problem in (\ref{utility}) is  convex. Furthermore, since the limiting point $\mathbf{F}(\infty)$ of algorithm \ref{wmmsealg} is  a stationary point of the problem (\ref{utility}), which is also the unique global optimal point of (\ref{utility}).

\section*{Appendix G: Proof of Theorem \ref{perfgap}}
Following the notation of the \emph{Bellman operators}  in (\ref{beloperator1})--(\ref{zerofunc}) in Appendix B, we define two mappings:
\bs $T_{\boldsymbol{\chi}}^\dagger( V, \mathbf{F}, \mathbf{U})  = T_{\boldsymbol{\chi}}^\dagger(\theta, V, \mathbf{F}, \mathbf{U}) + \theta =  {c}\left(\mathbf{Q}, \mathbf{F}\right)+  \sum_{k=1}^K \frac{\partial V \left(\mathbf{Q}  \right) }{\partial Q_k} \left[  R_k\left(\mathbf{H},  \mathbf{F}, \mathbf{U}\right) -\mu_k  \right]$,  $T_{\boldsymbol{\chi}}(V, \mathbf{F}, \mathbf{U})=T_{\boldsymbol{\chi}}^\dagger(V, \mathbf{F}, \mathbf{U})+\tau  G_{\boldsymbol{\chi}}(V,\mathbf{F}, \mathbf{U})$.\bsc

We  calculate the performance under policy $\widetilde{\Omega}$ as follows:
\begin{small}\begin{align}
	&\hspace{-0.2cm}\tilde{\theta}\tau= \mathbb{E}^{\widetilde{\Omega}}\big[\mathbb{E}\left[{c}\big(\mathbf{Q}, \widetilde{\Omega}\left(\boldsymbol{\chi}  \right)\big)\tau\big]\Big| \mathbf{Q}\right]\notag \\
	&\hspace{-0.6cm}\overset{(a)}=\mathbb{E}^{\widetilde{\Omega}}\bigg[\mathbb{E}\bigg[{c}\big(\mathbf{Q}, \widetilde{\Omega}\left(\boldsymbol{\chi}\right)\big)\tau+\sum_{\mathbf{Q} '} {\Pr}\big[ \mathbf{Q} '| (\boldsymbol{\chi},  \widetilde{\Omega}\left((\boldsymbol{\chi}  \right)\big]\widetilde{V} \left(\mathbf{Q}  '\right)  - \widetilde{V} \left(\mathbf{Q}  \right)   \Big| \mathbf{Q}\bigg]\bigg] 	\notag \\
	&\hspace{-0.6cm}\overset{(b)}= \mathbb{E}^{\widetilde{\Omega}}\left[\mathbb{E}\bigg[ {c}\big(\mathbf{Q}, \widetilde{\Omega}\left(\boldsymbol{\chi} \right)\big)\tau+\sum_{k=1}^K  \frac{\partial \widetilde{V}\left(\mathbf{Q} \right)}{\partial Q_k} \left[ R_k\big(\mathbf{H},\widetilde{\Omega}\left(\boldsymbol{\chi}\right) \big) -\mu_k  \right]\tau  \Big| \mathbf{Q}\bigg] \right. \notag \\
	&\hspace{4.5cm} \left. + \tau^2  G_{\mathbf{Q} }(\widetilde{V},\widetilde{\Omega}(\mathbf{Q} ))   \right] 	\label{finalexpr}
\end{align}\end{small}where \bs${\Pr}\big[ \mathbf{Q} '|\boldsymbol{\chi},  \widetilde{\Omega}\left(\boldsymbol{\chi}\right)\big]$\bsc is the discrete time transition kernel under policy $\widetilde{\Omega}$. (a) is due to \bs$ \mathbb{E}^{\widetilde{\Omega}}\big[\sum_{\mathbf{Q} '} \mathbb{E} \left[{\Pr}\big[ \mathbf{Q} '| \boldsymbol{\chi},  \widetilde{\Omega}\left(\mathbf{Q}  \right)\big]\big| \mathbf{Q}\right]\widetilde{V} \left(\mathbf{Q}  '\right)   \big]=\mathbb{E}^{\widetilde{\Omega}}\big[\mathbb{E}^{\widetilde{\Omega}}\big[\widetilde{V}(\mathbf{Q} ')\big|\mathbf{Q}  \big]\big]=\mathbb{E}^{\widetilde{\Omega}}\big[\widetilde{V}(\mathbf{Q} ) \big]$\bsc under the steady state distribution using $\widetilde{\Omega}$, and $(b)$ is due to the Taylor expansion of $\widetilde{V} \left(\mathbf{Q}  '\right) $ at $\widetilde{V} \left(\mathbf{Q}  \right) $.

Let $\Omega^\ast$ be the optimal policy solving the discrete time Bellman equation in (\ref{OrgBel}), then we have
\bs\begin{align}
	\mathbb{E}\left[T_{\boldsymbol{\chi}}({V^\ast}, \Omega^\ast(\boldsymbol{\chi}) \big|\mathbf{Q}\right]= \theta^\ast, \quad \forall \mathbf{Q} 	\label{asdadadasd}
\end{align}\bsc
Furthermore, according to the asymptotic optimality of Algorithm \ref{wmmsealg} in Lemma \ref{globalopt}, we have
\bs\begin{align}	\label{minach}
	T_{\boldsymbol{\chi}}^\dagger( \widetilde{V}, \widetilde{\Omega}(\boldsymbol{\chi} ))= \min_{\Omega(\mathbf{Q} )} T_{\boldsymbol{\chi}}^\dagger( \widetilde{V}, \Omega(\boldsymbol{\chi})), \quad \forall \boldsymbol{\chi}	
\end{align}\bsc
for sufficient small $L$. Dividing $\tau$ on both sizes of (\ref{finalexpr}), we obtain
\begin{align}
	\hspace{-0.3cm} \tilde{\theta}&=\mathbb{E}^{\widetilde{\Omega}}\big[\mathbb{E}[T_{\boldsymbol{\chi}}(\widetilde{V}, \widetilde{\Omega}(\boldsymbol{\chi} ))  |\mathbf{Q}]\big] \notag \\
	&\hspace{-0.3cm}=\mathbb{E}^{\widetilde{\Omega}}\big[\mathbb{E}[T_{\boldsymbol{\chi}}^\dagger(\widetilde{V}, \widetilde{\Omega}(\boldsymbol{\chi} ))+ \tau  G_{\boldsymbol{\chi}}(\widetilde{V},\widetilde{\Omega}(\boldsymbol{\chi} ))   |\mathbf{Q}]\big] 	\notag \\
	&\hspace{-0.3cm}\overset{(c)} \leq \mathbb{E}^{\widetilde{\Omega}}\big[\mathbb{E}[T_{\boldsymbol{\chi}}^\dagger(\widetilde{V}, {\Omega^\ast}(\boldsymbol{\chi} )) + \tau  G_{\boldsymbol{\chi}}(\widetilde{V},\widetilde{\Omega}(\boldsymbol{\chi} ))  |\mathbf{Q}]  \big] 	\notag \\
	&\hspace{-0.3cm}=\mathbb{E}^{\widetilde{\Omega}}\big[\mathbb{E}[T_{\boldsymbol{\chi}}(\widetilde{V}, \Omega^\ast(\boldsymbol{\chi} )) + \tau  G_{\boldsymbol{\chi}}(\widetilde{V},\widetilde{\Omega}(\boldsymbol{\chi} )) -\tau  G_{\boldsymbol{\chi}}(\widetilde{V},{\Omega^\ast}(\boldsymbol{\chi} ))  |\mathbf{Q}]  \big] 	\notag \\
	&\hspace{-0.3cm}\overset{(d)}=\mathbb{E}^{\widetilde{\Omega}}\big[\mathbb{E}[T_{\boldsymbol{\chi}}(\widetilde{V}, \Omega^\ast(\boldsymbol{\chi} )) - T_{\boldsymbol{\chi}}({V^\ast}, \Omega^\ast(\boldsymbol{\chi} )) + \theta^\ast \notag \\
	&\hspace{-0.3cm}+ \tau  G_{\boldsymbol{\chi}}(\widetilde{V},\widetilde{\Omega}(\boldsymbol{\chi} )) -\tau  G_{\boldsymbol{\chi}}(\widetilde{V},{\Omega^\ast}(\boldsymbol{\chi} ))   |\mathbf{Q} \big] 	\label{finaleeeee}
\end{align}
where $(c)$ is due to (\ref{minach}) and (d) is due to (\ref{asdadadasd}). For any given $\boldsymbol{\chi}$, since $G_{\boldsymbol{\chi}}$ is a smooth and bounded function, we have \bs$\tau  G_{\boldsymbol{\chi}}(\widetilde{V},\widetilde{\Omega}(\boldsymbol{\chi})) -\tau  G_{\boldsymbol{\chi}}(\widetilde{V},{\Omega^\ast}(\boldsymbol{\chi})) =\mathcal{O}(\tau)$\bsc. Therefore, from (\ref{finaleeeee}), we have \bs$\tilde{\theta} -  \theta^\ast  \leq \mathbb{E}^{\widetilde{\Omega}}\big[\mathbb{E}[T_{\boldsymbol{\chi}}(\widetilde{V}, \Omega^\ast(\boldsymbol{\chi})) - T_{\boldsymbol{\chi}}({V^\ast}, \Omega^\ast(\boldsymbol{\chi})) |\mathbf{Q} \big] + \mathcal{O}(\tau)  \overset{(e)}\leq  \alpha \| \mathbf{V}^\ast - \tilde{\mathbf{V}} \|_{\infty}^{\boldsymbol{\omega}} +\mathcal{O}(\tau) \overset{(f)}\leq  o(1)+\mathcal{O}(L)+\mathcal{O}(\tau)  =o(1)+\mathcal{O}(L)$\bsc, where $\alpha>0$ is some constant, \bs$\mathbf{V}^\ast =\{V^\ast(\mathbf{Q} ): \forall \mathbf{Q} \}$ and $\widetilde{\mathbf{V}} =\{\widetilde{V}(\mathbf{Q} ): \forall \mathbf{Q} \}$\bsc. $(e)$  holds under some sup-norm $\| \cdot \|_{\infty}^{\boldsymbol{\omega}} $, which is due to the Lipchitz continuity of the operator $T_{\boldsymbol{\chi}}$ \cite{DP_Bertsekas} and $(f)$ is due to \bs$V(\mathbf{Q} )-\widetilde{V}(\mathbf{Q} )=o(1)+\mathcal{O}(L)$\bsc.

\end{document}